\documentclass[letterpaper,onecolumn,11pt,accepted=2019-04-24]{quantumarticle}
\usepackage{fullpage}

\pdfoutput=1
\usepackage[numbers,sort&compress]{natbib}
\usepackage[utf8]{inputenc}
\usepackage[english]{babel}
\usepackage[T1]{fontenc}
\usepackage{amsmath}
\usepackage{hyperref}

\usepackage{amsfonts}
\usepackage{amsthm}
\usepackage{amssymb}
\usepackage{braket}
\usepackage{latexsym}
\usepackage[usenames,dvipsnames]{color}
\usepackage{soul}
\usepackage{paralist}
\usepackage{bm}

\newtheorem{theorem}{Theorem}[section]

\newtheorem{lemma}[theorem]{Lemma}

\newtheorem{definition}[theorem]{Definition}

\makeatletter
\newtheorem*{rep@theorem}{\rep@title}
\newcommand{\newreptheorem}[2]{%
\newenvironment{rep#1}[1]{%
 \def\rep@title{#2 \ref{##1}}%
 \begin{rep@theorem}}%
 {\end{rep@theorem}}}
\makeatother

\newreptheorem{theorem}{Theorem}

\usepackage{graphicx}

\newcommand{\comment}[1]{}
\newcommand{\pr}{\operatorname{Pr}}
\newcommand{\tinyspace}{\mspace{1mu}}

\newcommand{\abs}[1]{\left\lvert\tinyspace #1 \tinyspace\right\rvert}

\newcommand{\norm}[1]{\left\lVert\tinyspace #1 \tinyspace\right\rVert}

\newcommand{\tr}{\operatorname{Tr}}

\newcommand{\setft}[1]{\mathrm{#1}}
\newcommand{\lin}[1]{\setft{L}\left(#1\right)}
\newcommand{\density}[1]{\setft{D}\left(#1\right)}

\newcommand{\class}[1]{\textup{#1}}
\newcommand{\prob}[1]{\textsc{#1}}
\newcommand{\ayes}{A_{\rm yes}}
\newcommand{\ano}{A_{\rm no}}
\DeclareMathOperator*{\argmax}{arg\,max}

\newcommand{\ketbra}[2]{\ket{#1}\!\bra{#2}}
\newcommand{\brakett}[2]{\mbox{$\langle #1  | #2 \rangle$}}

\newcommand{\enorm}[1]{\norm{#1}_{\mathrm{2}}}      
\newcommand{\trnorm}[1]{\norm{#1}_{\mathrm {tr}}}  
\newcommand{\snorm}[1]{\norm{#1}_{\mathrm {\infty}}}    

\newcommand{\hin}{H_{\rm in}}
\newcommand{\hprop}{H_{\rm prop}}
\newcommand{\hout}{H_{\rm out}}
\newcommand{\hstab}{H_{\rm stab}}

\newcommand{\app}{\prob{APX-SIM}}

\newcommand{\tpc}{\prob{APX-2-CORR}}
\newcommand{\spgap}{\prob{SPECTRAL-GAP}}
\newcommand{\spa}[1]{\mathcal{#1}}

\newcommand{\trace}{\tr}

\newcommand{\hw}{{\rm HW}}

\newcommand{\poly}{\textup{poly}}

\newcommand{\mypar}[1]{\vspace{2mm}\noindent\emph{#1.}}

\def\({\left(}
\def\){\right)}

\newcommand{\complex}{{\mathbb C}}
\newcommand{\reals}{{\mathbb R}}

\newcommand{\QMA}{\class{QMA}}
\newcommand{\coQMA}{\class{co-QMA}}
\newcommand{\UQMA}{\class{UQMA}}

\newcommand{\BQP}{\class{BQP}}

\newcommand{\PQMA}{\class{P}^{\class{QMA}[\class{log}]}}
\newcommand{\PNP}{\class{P}^{\class{NP}[\class{log}]}}
\newcommand{\PUQMA}{\class{P}^{\class{UQMA}[\class{log}]}}
\newcommand{\psihist}{\psi_{\rm hist}}
\mathchardef\mhyphen="2D

\begin{document}

\title{The complexity of simulating local measurements on quantum systems}

\author[1,3]{Sevag Gharibian}
\email{sgharibian@upb.de}
\homepage{http://groups.uni-paderborn.de/fg-qi/people.html}
\orcid{0000-0002-9992-3379}

\author[2,3]{Justin Yirka}
\email{yirka@utexas.edu}
\homepage{https://www.justinyirka.com/}
\orcid{0000-0001-6173-2465}

\affil[1]{Department of Computer Science, Paderborn University, Germany}
\affil[2]{Department of Computer Science, The University of Texas at Austin, USA}
\affil[3]{Department of Computer Science, Virginia Commonwealth University, USA}

\date{} 

\maketitle

\begin{abstract}
An important task in quantum physics is the estimation of local quantities for ground states of local Hamiltonians. Recently, [Ambainis, CCC 2014] defined the complexity class $\PQMA$, and motivated its study by showing that the physical task of estimating the expectation value of a local observable against the ground state of a local Hamiltonian is $\PQMA$-complete.
In this paper, we continue the study of $\PQMA$, obtaining the following lower and upper bounds.\vspace{2mm}

\noindent Lower bounds (hardness results):
    \begin{itemize}
        \item The $\PQMA$-completeness result of [Ambainis, CCC 2014] requires $O(\log n)$-local observables and Hamiltonians. We show that simulating even a \emph{single qubit} measurement on ground states of $5$-local Hamiltonians is $\PQMA$-complete, resolving an open question of Ambainis.
    \item We formalize the complexity theoretic study of estimating two-point correlation functions against ground states, and show that this task is similarly $\PQMA$-complete.
    \item We identify a flaw in [Ambainis, CCC 2014] regarding a $\PUQMA$-hardness proof for estimating spectral gaps of local Hamiltonians. By introducing a ``query validation'' technique, we build on [Ambainis, CCC 2014] to obtain $\PUQMA$-hardness for estimating spectral gaps under polynomial-time Turing reductions.
    \end{itemize}
\noindent Upper bounds (containment in complexity classes):
    \begin{itemize}
        \item $\PQMA$ is thought of as ``slightly harder'' than QMA. We justify this formally by exploiting the hierarchical voting technique of [Beigel, Hemachandra, Wechsung, SCT 1989] to show $\PQMA\subseteq \class{PP}$. This improves the containment $\class{QMA}\subseteq \class{PP}$ [Kitaev, Watrous, STOC 2000].
    \end{itemize}
    This work contributes a rigorous treatment of the subtlety involved in studying oracle classes in which the oracle solves a \emph{promise} problem. This is particularly relevant for {quantum} complexity theory, where most natural classes such as BQP and QMA are defined as promise classes.
\end{abstract}

\tableofcontents

\section{Introduction}\label{scn:intro}

\subsection{Background on Quantum Hamiltonian Complexity}\label{sscn:background}
The use of computational complexity theory to study the inherent difficulty of computational problems has proven remarkably fruitful over the last decades. For example, the theory of NP-completeness~\cite{C72,L73,K72} has helped classify the worst-case complexity of hundreds of computational problems which elude efficient classical algorithms. In the quantum setting, the study of a quantum analogue of NP, known as Quantum Merlin Arthur\footnote{More precisely, QMA is Merlin-Arthur (MA) with a quantum proof and quantum verifier.} (QMA), was started in 1999 by the seminal ``quantum Cook-Levin theorem'' of Kitaev~\cite{KSV02}, which showed that estimating the ground state energy of a given $k$-local Hamiltonian is QMA-complete for $k\geq 5$. Here, a $k$-local Hamiltonian $H$ can be thought of as a quantum constraint satisfaction system in which each quantum clause acts non-trivially on $k$ qubits. More formally, $H\in\complex^{2^n\times 2^n}$ is an exponentially large Hermitian matrix acting on $n$ qubits, but with a succinct description\footnote{Implicitly, if $H_i$ acts on a subset $S_i\subseteq[n]$ of qubits non-trivially, then more accurately one writes $H_i\otimes I_{[n]\setminus S_i}$. We write $H=\sum_i H_i$ for simplicity.} $H=\sum_i H_i$, where each local clause $H_i\in\complex^{2^k\times 2^k}$ acts non-trivially on $k$ qubits. The ``largest total weight of satisfiable clauses'' is given by the \emph{ground state energy} of $H$, i.e. the smallest eigenvalue of $H$. Physically, the ground state energy and its corresponding eigenvector, the \emph{ground state}, are motivated in that they represent the energy level and state of a given quantum system at low temperature, respectively. For this reason, since Kitaev's work~\cite{KSV02}, a number of physically motivated problems have been shown complete for QMA (this has given rise to the field of \emph{Quantum Hamiltonian Complexity}, see, e.g.,~\cite{O11},~\cite{Boo14} and~\cite{GHLS15} for surveys\footnote{Since these surveys were published, the study of ground spaces through the field of {Quantum Hamiltonian Complexity} has continued to grow. For example, recent variants on circuit-to-Hamiltonian mappings such as the space-time construction of Breuckmann and Terhal~\cite{BT14}, of Bausch and Crosson~\cite{BC18} for considering complex weights and branching transitions, of Bausch, Cubitt and Ozols~\cite{BCO17} for embedding $\class{QMA}_{\class{exp}}$ computations into 1D translation invariant systems of local dimension $42$, or of Caha, Landau, and Nagaj~\cite{CLN18} for achieving an arbitrarily high success probability of extracting a computation result from a history state while only needing to increase the clock size logarithmically, have been given. Bravyi and Hastings~\cite{BH17} have shown that the local Hamiltonian problem for the quantum Ising model is StoqMA-complete, completing the complexity classification scheme of Cubitt and Montanaro~\cite{CM16}. Gily\'{e}n and Sattath~\cite{GS17} have given a constructive Quantum Lovasz Local Lemma which efficiently prepares a frustration-free Hamiltonian's ground state under the assumption that the system is ``uniformly'' gapped. (This list of results is a small sample of recent work.)}). Many of these QMA-complete problems focus on estimating ground state energies of local Hamiltonians.

\mypar{Beyond ground state energies} In recent years, however, new directions in quantum complexity theory involving other physical properties of local Hamiltonians have appeared. We attempt to survey a number of such results here.

Brown, Flammia and Schuch~\cite{BFS11} (also Shi and Zhang~\cite{SZ}) introduced a quantum analogue of $\class{\#P}$, denoted $\class{\#BQP}$, and showed that computing the ground state degeneracy or density of states of local Hamiltonians is $\class{\#BQP}$-complete. Gharibian and Sikora~\cite{GS18} showed that determining whether the ground space of a local Hamiltonian has an ``energy barrier'' is QCMA-complete, where QCMA~\cite{AN02} is Merlin-Arthur (MA) with a classical proof and quantum prover. This was strengthened by Gosset, Mehta, and Vidick~\cite{GMV17}, who showed that QCMA-completeness holds even for \emph{commuting} local Hamiltonians. Bravyi and Gosset~\cite{BG17} studied the complexity of quantum impurity problems, which involve a bath of free fermions coupled to an interacting subsystem dubbed an ``impurity''. Cubitt, Montanaro, and Piddock~\cite{CMP18} have shown that certain simple spin-lattice models are ``universal'', in the sense that they can replicate the entire physics of any other quantum many-body system.

From a hardness of approximation perspective, Gharibian and Kempe~\cite{GK12} introduced cq-${\rm\Sigma_2}$, a quantum generalization of $\Sigma_2^p$, and showed that determining the smallest subset of interaction terms of a given local Hamiltonian which yields a frustrated ground space is cq-${\rm\Sigma_2}$-complete (and additionally, cq-${\rm\Sigma_2}$-hard to approximate). Aharonov and Zhou~\cite{AZ18} studied the task of ``Hamiltonian sparsification'' or ``degree-reduction'', in which one attempts to simulate an input local Hamiltonian with a new local Hamiltonian with an interaction graph of bounded degree while preserving only the ground space and spectral gap (in general,~\cite{AZ18} show this is impossible). This was in pursuit of answering whether classical proof techniques for the classical PCP theorem carry over to the quantum setting.

Finally, various authors have studied spectral gaps of local Hamiltonians from a \emph{computability theory} perspective. (In computability theory, one asks whether a decision or promise problem can be decided by a Turing machine running in a finite number of steps. In the current paper, our focus is instead on complexity theory, in which problems are typically known to be computable; the question is rather to obtain an estimate of the resources the Turing machine requires to solve the problem.) Gosset and Mozgunov~\cite{GM16} and Bravyi and Gosset~\cite{BG15} have studied spectral gaps for frustration-free 1D translation-invariant systems (the latter, in particular, shows that distinguishing between gapped and gapless phases is decidable in the spin-1/2 chains studied). Cubitt, Perez-Garcia and Wolf~\cite{CPW15} have shown undecidability of estimating spectral gaps in the thermodynamic limit for translation-invariant, nearest-neighbor Hamiltonians on a 2D square lattice. This has very recently been improved to undecidability for 1D translation invariant systems by Bausch, Cubitt, Lucia, and Perez-Garcia~\cite{BCLP18}. Note that both the current paper and~\cite{A14} also study spectral gaps, but from a complexity theory perspective; in particular, in contrast to the undecidability studies listed above, which consider the thermodynamic limit (i.e. the number of qubits $n$ goes to infinity), in our setting the number of qubits $n$ (expressed in unary) is part of the input to the problem, as is standard in complexity theory.

\subsection{Simulating local measurements on low-temperature quantum systems}
In Section~\ref{sscn:background}, we listed a number of results studying properties of local Hamiltonians beyond ground state energies. We intentionally omitted one particular problem, however, which is extremely well motivated physically and which is the starting point of this work. Suppose one cools a many-body quantum system to low temperature in the lab and \emph{a priori} does not know the system's properties (such as its ground state energy, ground space, etc); this is a natural assumption, as many of these properties are $\QMA$-complete to estimate to begin with. The most ``basic'' action an experimenter can now take is to perform a local measurement (typically on a constant number of qubits) in an attempt to extract information about the system. The question is: How difficult is it to computationally simulate this ``basic'' action?

This computational task was formalized by Ambainis~\cite{A14} (formal definitions in Section~\ref{scn:preliminaries}) and designated \emph{Approximate Simulation (APX-SIM)}: \emph{Given a $k$-local Hamiltonian $H$ and an $l$-local observable $A$, estimate the expectation value of the measurement $A$ against the ground state of $H$, i.e. estimate $\langle A\rangle:=\bra{\psi}A\ket{\psi}$ for $\ket{\psi}$ a ground state of $H$}. It turned out that not only is this problem ``hard'', but that it is in fact harder then even QMA.

To formalize this, Ambainis introduced~\cite{A14} the complexity class $\PQMA$, which intuitively is ``slightly harder'' than QMA, and showed that \app\ is $\PQMA$-complete when the Hamiltonian $H$ and observable $A$ are both $O(\log n)$-local. To help set context before stating our results, we now discuss $\PQMA$ and its classical analogue $\PNP$ in further depth.

\subsection{Oracle complexity classes, $\PQMA$, and $\PNP$}\label{sscn:oracles}

Formally, $\PQMA$ is the class of decision problems which can be decided by a polynomial-time deterministic Turing machine making $O(\log n)$ queries to an oracle for QMA. Thus, it is an example of an \emph{oracle complexity class}.

Let us make three remarks. First, by ``oracle for QMA'', one typically means an oracle for a QMA-complete problem $\Pi$ (since any other problem in QMA may be reduced in polynomial-time to $\Pi$); in this paper, we shall set $\Pi$ as the QMA-complete $2$-local Hamiltonian problem ($2$-LH)~\cite{KKR06}, an instance $(H,a,b)$ of which asks: \emph{Is the ground state energy of $H$ at most $a$ (YES case), or at least $b$ (NO case), for $b-a\geq1/\poly(n)$?} Second, although $\PQMA$ uses a QMA oracle, it in fact also contains the complementary class co-QMA. This is because to solve any co-QMA problem, a $\PQMA$ machine can plug the co-QMA problem instance into the QMA oracle, and subsequently flip the oracle's answer in polynomial-time to solve the co-QMA problem. Thus, $\PQMA\neq\QMA$ unless $\coQMA\subseteq\QMA$ (which is unlikely), and so $\PQMA$ is likely strictly harder than QMA. The third remark is that $\PQMA$ uses an oracle for a class of \emph{promise} problems; this is a subtle but crucial point, which we discuss shortly; first, let us set the stage by reviewing the analogous classical class $\PNP$.

\mypar{The class $\PNP$} $\PNP$ is defined analogously to $\PQMA$ except it utilizes an NP oracle. Analogous to $\PQMA$, $\PNP$ contains both NP and co-NP and thus is likely strictly harder than NP. As an upper bound, $\PNP\subseteq \class{NP}^{\class{NP}}=\Sigma_2^p$, where an $\class{NP}^{\class{NP}}$ machine is an NP machine which nondeterministically makes up to a polynomial number of calls to an NP oracle, and $\Sigma_2^p$ is the second level of the Polynomial-Time Hierarchy (PH). In contrast, it is unlikely for $\PQMA$ to be in PH, as even $\BQP\subseteq\QMA\subseteq\PQMA$ is generally not believed to be in PH \cite{Aa10,FSUV12,C16,R16,RT18}. A natural question is whether $\PQMA$ might instead be contained in an appropriate quantum analogue of PH. The answer is not clear, since unlike the fact that $\Sigma_2^p=\class{NP}^{\class{NP}}$, it is not clear that ``quantum $\Sigma_2^p$'' should equal (say) $\class{QMA}^{\class{QMA}}$; see~\cite{GSSSY18} for a discussion and treatment of quantum analogues of PH. Additional references on ``quantum PH'' are~\cite{Y02,GK12,LG17}.

As far as we are aware, whether $\class{P}^{\class{NP}[k]}$ (i.e. $k\in O(1)$ NP queries), $\class{P}^{\class{NP}[\log^k n]}$ (i.e. $O(\log^k n)$ NP queries for $k\geq 1$), and $\class{P}^{\class{NP}}$ (i.e. polynomially many NP queries) coincide remains open. Notably, if $\class{P}^{\class{NP}[1]}=\class{P}^{\class{NP}[2]}$, then $\class{P}^{\class{NP}[1]}=\PNP$ and PH collapses~\cite{H87}. But, it \emph{is} known is that if the queries to the NP oracle are \emph{non-adaptive}, meaning they are all made in parallel, then the resulting class $\class{P}^{||\class{NP}}$ equals $\PNP$~\cite{BH91,H89}. (The analogous statement $\class{P}^{||\class{QMA}}=\PQMA$ has recently been shown~\cite{GPY19}.) Similarly, a P machine making $O(\log^k n)$ adaptive NP queries is equivalent in power to a P machine making $O(\log^{k+1} n)$ non-adaptive queries for all $k\geq 1$~\cite{CS05}. In terms of complete problems, determining the election winner in Lewis Carroll's 1876 voting system is $\PNP$-complete~\cite{HHR97}, and model checking for certain branching-time temporal logics is $\class{P}^{\class{NP}[\log^2 n]}$-complete~\cite{S03}.

Finally, it is important to note that NP is a class of \emph{decision} problems. Formally, this means any NP problem is specified by a language $L\subseteq\set{0,1}^*$ with the corresponding decision problem: Given input $x\in\set{0,1}^*$, is $x\in L$? In the context of $\PNP$, this means \emph{any} string $x\in\set{0,1}^*$ is a valid problem instance or \emph{query} to an  NP oracle for language $L$, since any such $x$ is either in the language or not. Unfortunately, an analogous statement cannot be made about $\PQMA$, complicating its study; this brings us to our third, crucial remark about $\PQMA$ from earlier.

\mypar{Oracles for promise classes and the issue of invalid queries} In contrast to NP, QMA is a class of \emph{promise} problems. Formally, this means the space of all inputs $\set{0,1}^*$ is partitioned into \emph{three} sets, $A,B,C$, where $A$ and $C$ are YES and NO instances, respectively, and $B$ is the set of ``invalid'' instances. (Classes of decision problems are a special case of this in which $A=L$, $B=\emptyset$, and $C=\set{0,1}^*\setminus L$.) Why might this pose a problem for studying $\PQMA$?

Recall that we assume all calls by the $\PQMA$ machine to the QMA oracle $Q$ are for instances $(H,a,b)$ of $2$-LH: \emph{Is the ground state energy of $H$ at most $a$ (YES case), or at least $b$ (NO case), for $b-a\geq1/\poly(n)$}? Unfortunately, a P machine cannot in general tell\footnote{By definition, the P machine can only execute polynomial-time computations. Thus, even if we assume without loss of generality that the P machine uses specific circuit-to-Hamiltonian constructions (for preparing queries to the QMA oracle) whose promise gaps are precisely known (e.g.~\cite{BC18,CLN18}), it is not clear how the machine should know whether the quantum circuit it feeds into said construction satisfies a promise gap to begin with --- where such circuit-to-Hamiltonian constructions require the input circuit to satisfy an inverse polynomial promise gap, which is unlikely to be relaxed to more easily verified, say, inverse exponentially small promise gaps, since it would imply $\class{PSPACE}\subseteq\class{QMA}$, which follows since QMA with exponentially small gap equals PSPACE~\cite{FL18}.} whether the instance $(H,a,b)$ it feeds to $Q$ satisfies the promise conditions of LH; in particular, the ground state energy may lie in the interval $(a,b)$. We call any such instance or query $(H,a,b)$ violating the $2$-LH promise ``invalid'' (these instances belong in the set $B$ of instances above). For any invalid query, the oracle $Q$ is allowed to accept or reject arbitrarily.  This raises a potential issue: if the oracle is allowed to respond arbitrarily to invalid queries, how does one ensure a YES instance (or NO instance) of a $\PQMA$ problem is well-defined (normally, the P machine conditions its future actions on the response the oracle returns for each query)? To do so, we stipulate (see, e.g., Definition 3 of Goldreich~\cite{G06}) that the P machine must output the \emph{same} answer regardless of how any invalid queries are answered by the oracle. We view the issue of formally handling invalid queries as one of the central themes and contributions of this work.

\subsection{Our results}\label{sscn:results}
Our results fall into two categories: \emph{Lower bounds} (hardness results) and \emph{upper bounds} (containment in complexity classes).

\paragraph{Lower bounds (hardness results).} We begin by showing three hardness results; two focus on computational problems introduced in~\cite{A14} (\app\ and \spgap) and one introduces a new problem (\tpc).\\

\noindent \emph{1. $\PQMA$-completeness of \app\ for $O(1)$-local Hamiltonians and single-qubit observables.} Recall that in~\cite{A14}, Ambainis introduced $\PQMA$ and showed that \app\ (i.e. given $k$-local Hamiltonian $H$ and $l$-local observable $A$, simulate measurement $A$ on the ground state of $H$) is $\PQMA$-complete. This proof required both the Hamiltonian $H$ and observable $A$ to be $O(\log n)$-local. From a physical standpoint, however, it is typically desirable to have $O(1)$-local Hamiltonians and observables --- whether $\PQMA$-hardness holds in this regime was left as an open question in~\cite{A14}. We thus first ask: \emph{Is \app\ still hard for $O(1)$-local Hamiltonians and $1$-local observables?}

Let us develop some intuition before answering this question. Typically, computational problems (such as estimating ground state energies) of $1$-local \emph{Hamiltonians} are easy, since the qubits do not interact. This intuition does {not}, however, carry over to the setting of simulating $1$-local \emph{measurements}. For example, by embedding a $3$-SAT instance $\phi$ into a $3$-local Hamiltonian and then using the ability to repeatedly measure observable $Z$ against single qubits of the ground state, we can extract a solution to $\phi$! Thus, the $3$-local Hamiltonian and $1$-local observable case is at least NP-hard. Indeed, here we show it is much harder, resolving Ambainis's open question.

\begin{theorem}\label{thm:main1}
    Given a $5$-local Hamiltonian $H$ on $n$ qubits and a $1$-local observable $A$, estimating $\langle A\rangle $ (i.e. \app) is $\PQMA$-complete.
\end{theorem}

\noindent Thus, measuring just a \emph{single} qubit of the ground state of a local Hamiltonian $H$ with a single-qubit observable $A$ (in fact, $A$ is fixed, independent of $H$ in our construction) is harder than QMA (assuming $\class{QMA}\neq \PQMA$, which is likely as otherwise $\class{co-QMA}\subseteq\class{QMA}$).\\

\noindent \emph{2. $\PQMA$-completeness of \tpc\ for $O(1)$-local Hamiltonians and single-qubit observables.} We next introduce a second natural problem related to \app, denoted \tpc. \tpc\ is defined similarly to $\app$ except one is given Hamiltonian $H$ and observables $A$ and $B$ and asked to estimate the \emph{two-point correlation function} $\langle A\otimes B\rangle -\langle A \rangle\langle B\rangle$ (recall $\langle A\rangle:=\bra{\psi}A\ket{\psi}$ for $\ket{\psi}$ a ground state of $H$). By modifying the construction behind the proof of Theorem~\ref{thm:main1}, we also show \tpc\ is $\PQMA$-complete.

\begin{theorem}\label{thm:main2}
    Given a $5$-local Hamiltonian $H$ on $n$ qubits and a pair of $1$-local observables $A$ and $B$, estimating $\langle A\otimes B\rangle -\langle A \rangle\langle B\rangle$ (i.e. \tpc) is $\PQMA$-complete.
\end{theorem}

\noindent\emph{3. $\PUQMA$-hardness of estimating spectral gaps.} A third well-motivated problem involving Hamiltonians is \spgap~\cite{A14}: \emph{Given a $k$-local Hamiltonian $H$, estimate $\lambda_1(H)-\lambda_2(H)$, where $\lambda_i(H)$ denotes the $i$-th smallest eigenvalue of $H$}. (For clarity, if the ground space of $H$ is degenerate, we define its spectral gap as $0$.) In~\cite{A14}, it is shown that $\spgap\in\PQMA$, and a claimed proof is given that \spgap\ for $O(\log n)$-local Hamiltonians $H$ is $\PUQMA$-hard. Here, $\PUQMA$ is $\PQMA$ except with a Unique QMA oracle, and Unique QMA is roughly QMA with a unique accepting quantum witness in the YES case (see Section~\ref{scn:spectralGap} for formal definitions).

Recall now that, as discussed in Section~\ref{sscn:oracles}, a central theme of this work is the subtlety involved in the study of oracle classes in which the oracle solves a \emph{promise} problem (such as $\PQMA$), as opposed to a decision problem (such as $\PNP$). In particular, for $2$-LH queries to the QMA oracle which violate the promise gap for $2$-LH, the oracle is allowed to give an arbitrary answer. We observe that this point appears to have been missed in~\cite{A14}'s claimed proofs of $\PQMA$-hardness of $\app$ and of $\PUQMA$-hardness of \spgap, rendering the proofs incorrect. Though, as shown by our first result, we are able to correct and improve on the proof that $\app$ is $\PQMA$-hard.

In our last result, we also overcome this difficulty to recover $\PUQMA$-hardness of \spgap, with the tradeoff that we obtain a ``slightly weaker'' hardness claim, in the sense that it follows from a Turing reduction as opposed to a Karp or mapping reduction\footnote{A Karp, many-one, or mapping reduction from a decision problem $A$ to a decision problem $B$ maps any input $\Pi_A$ for $A$ into an input $\Pi_B$ for $B$ such that $\Pi_A$ is a YES (NO) instance of $A$ if and only if $\Pi_B$ is a YES (NO) instance for $B$. A Turing reduction generalizes the notion of a Karp reduction; here, one is given access to an oracle solving $B$, and the goal is to solve $\Pi_A$ using potentially multiple calls to the oracle. In this paper, both notions of reduction are assumed to be deterministic polynomial-time.} (\cite{A14} claimed hardness under the latter). In the process, we also improve the locality of $H$ to $O(1)$.

\begin{theorem}\label{thm:spgap}
	Given a $4$-local Hamiltonian $H$, estimating its spectral gap (i.e. $\spgap$) is $\PUQMA$-hard under polynomial-time Turing reductions.
\end{theorem}
\noindent Thus, the ability to solve $\PUQMA$ problems in polynomial time would imply a polynomial-time algorithm for \spgap. Note that whether $\UQMA=\QMA$ remains open. For the cases of NP~\cite{VV86} or MA and QCMA~\cite{ABBS08}, it is known that Unique NP, Unique MA, and Unique QCMA reduce to NP, MA, and QCMA, respectively, under randomized reductions.

\paragraph{Upper bounds (containment in complexity classes).} As mentioned earlier, since both \QMA\ and \coQMA\ are contained in $\PQMA$, it is likely that $\PQMA$ is strictly more powerful than both \QMA\ and \coQMA. This raises the question: \emph{What is an upper bound on the power of $\PQMA$?} Two points of reference are worth mentioning here. First, the classical analogue of $\PQMA$, $\PNP$, is known to be upper bounded by PP~\cite{BHW89}. Here, PP is the set of promise problems solvable in probabilistic polynomial time with \emph{unbounded} error; in other words, a PP machine accepts YES instances with probability strictly larger than $1/2$ and accepts NO instances with probability less than or equal to $1/2$. The second point of reference is that $\class{QMA}$ is bounded by $\class{PP}$~\cite{KW00,MW05} (\cite{Vy03} actually shows the slightly stronger containment $\class{QMA}\subseteq \class{A}_0\class{PP}$). We thus ask whether, like its classical analogue, $\PQMA$ can bounded by PP, and so is ``slightly harder'' than $\QMA$. Indeed, we show this is the case.
\begin{theorem}\label{thm:inPP}
    $\PQMA\subseteq\class{PP}$.
\end{theorem}
\noindent Thus, $\QMA\subseteq\PQMA\subseteq\class{PP}$, which rigorously justifies the intuition that $\PQMA$ should be thought of as ``slightly harder'' than $\QMA$.\\

\subsection{Proof techniques} We now outline our proof techniques for both our upper bound and lower bound results.

\paragraph{Lower bounds.} We begin with proof techniques techniques for our lower bounds involving \app\ (Theorem~\ref{thm:main1}), \tpc\ (Theorem~\ref{thm:main2}), and \spgap\ (Theorem~\ref{thm:spgap}).\\

\noindent \emph{1. $\PQMA$-hardness of $\app$.} To show Theorem~\ref{thm:main1}, our intuition is simple: To design our local Hamiltonian $H$ so that its ground state encodes a so-called history state\footnote{A \emph{history state} can be seen as a quantum analogue of the ``tableau'' which appears in the proof of the Cook-Levin theorem, i.e. a history state encodes the history of a quantum computation. In contrast to tableaus, however, the history encodes information in quantum {superposition}.} $\ket{\psi}$ for a given $\PQMA$ instance such that measuring observable $Z$ on the designated ``output qubit'' of $\ket{\psi}$ reveals the answer of the computation. At a high level, this is achieved by combining a variant of Kitaev's circuit-to-Hamiltonian construction~\cite{KSV02} (which forces the ground state to follow the P circuit) with Ambainis's ``query Hamiltonian''~\cite{A14} (which forces the ground state to encode correctly answered queries to the QMA oracle). Making this rigorous requires a careful analysis of the ground space of Ambainis's query Hamiltonian when queries violating the promise gap of the oracle are allowed (Lemma~\ref{l:amborig}), showing a simple but useful extension of Kempe, Kitaev, and Regev's Projection Lemma~\cite{KKR06} (Lemma~\ref{cor:kkr}) that any low energy state of $H$ must be close to a valid history state, which thus allows us to use just a single-qubit observable, and applying Kitaev's unary encoding trick~\cite{KSV02} to bring the locality of the Hamiltonian $H$ down to $O(1)$ (Lemma~\ref{l:amb}).\\

 \noindent \emph{2. Containment of \tpc\ in $\PQMA$.} The hardness proof for $\tpc$ is similar to that of $\app$, so we focus instead on the containment proof of $\tpc\in\PQMA$, for which a trick is required. For containment, the naive approach would be to run Ambainis's $\PQMA$ protocol~\cite{A14} for \app\ independently for each term $\langle A\otimes B\rangle$, $\langle A\rangle$, and $\langle B\rangle$. However, if a cheating prover does not send the \emph{same} ground state $\ket{\psi}$ for each of these measurements, soundness of the protocol can be violated.

 To address this, we exploit a trick of Chailloux and Sattath~\cite{CS11} from the setting of QMA(2): we observe that the correlation function requires only knowledge of the two-body reduced density matrices $\set{\rho_{ij}}$ of $\ket{\psi}$. A prover can send classical descriptions of the $\set{\rho_{ij}}$ (which is possible since each $\rho_{ij}$ is of constant size), along with a ``consistency proof'' for the QMA-complete Consistency problem~\cite{L06} to ensure the $\set{\rho_{ij}}$ indeed correspond to some state. The verifier can then freely copy each $\rho_{ij}$ and thus can calculate each term.\\

\noindent\emph{3. $\PUQMA$-hardness of estimating spectral gaps.} As mentioned in Sections~\ref{sscn:oracles} and~\ref{sscn:results}, a central theme of this work is the fact that a $\PQMA$ machine can make invalid queries to the QMA oracle, and this point appears to have been missed in~\cite{A14}, where all queries were assumed to satisfy the LH promise. This results in the proofs of two key claims of~\cite{A14} being incorrect. The first claim was used in the proof of $\PQMA$-completeness for \app\ (Claim 1 in~\cite{A14}); we provide a corrected statement and proof in Lemma~\ref{l:amborig} (which suffices for the $\PQMA$-hardness results in~\cite{A14} regarding \app\ to hold).

The error in the second claim (Claim 2 of~\cite{A14}), wherein $\PUQMA$-hardness of determining the spectral gap of a local Hamiltonian is shown, appears arguably more serious. The construction of~\cite{A14} requires a certain ``query Hamiltonian'' to have a spectral gap, which indeed holds if the $\PUQMA$ machine makes no invalid queries. However, if the machine makes invalid queries, this gap can close, and it is not clear how one can recover $\PUQMA$-hardness under mapping reductions.

To overcome this, we introduce a technique of ``query validation'': given a query to the $\UQMA$ oracle, we would like to determine if the query is valid or ``far'' from valid. While it is not clear how a P machine alone can perform such ``query validation'', we show how to use a $\spgap$ oracle to do so, allowing us to eliminate ``sufficiently invalid'' queries. Combining this idea with Ambainis's original construction~\cite{A14}, we show Theorem~\ref{thm:spgap}, i.e. $\PUQMA$-hardness for $\spgap$, with the additional improvement that the input Hamiltonians are now $O(1)$-local, versus $O(\log n)$-local as in~\cite{A14}, via the same unary encoding trick from our previous results. Since our ``query validation'' requires a polynomial number of calls to the $\spgap$ oracle, this result requires a polynomial-time \emph{Turing} reduction. Whether this can be improved to a mapping reduction is left as an open question.

\paragraph{Upper bounds.} We now move to our upper bound result, which is the most technically involved. To show  $\PQMA\subseteq\class{PP}$ (Theorem~\ref{thm:inPP}), we exploit the technique of \emph{hierarchical voting}, used by Beigel, Hemachandra, and Wechsung~\cite{BHW89} to show $\PNP\subseteq \class{PP}$, in conjunction with the QMA strong amplification results of Marriott and Watrous~\cite{MW05}.

To give some intuition as to how the technique works, we sketch its application in the classical case of $\PNP$~\cite{BHW89}. There, we have a P machine making queries to an NP oracle for SAT (i.e. the $i$-th query feeds the oracle some SAT formulae $\phi_i$), and the goal is to simulate this by a PP machine. Since the PP machine does not have access to an NP oracle, the best it can do is \emph{guess} the answer to each query; so, it begins by  guessing the answers to each NP query by picking a random assignment $x_i$ for each formula $\phi_i$ in the hope that $\phi_i(x_i)=1$. Of course, such a guess $x_i$ can be satisfying only if $\phi_i$ is satisfiable to begin with; with a bit of thought, this implies that the ``correct'' query string $y^*$ (i.e. the $i$-th bit of $y^*$ equals $1$ if and only if $\phi_i$ is satisfiable) is the lexicographically \emph{largest} string $y^*$ attainable by this guessing process.

Unfortunately, the PP machine might have a very small probability of guessing $y^*$; thus, it next ``boosts'' this probability via several rounds of ``hierarchical voting''. Conceptually, this technique does not actually \emph{increase} the probability of outputting $y^*$ but rather decreases the probability of outputting other, lexicographically smaller query strings $y'\neq y^*$. This works because the machine is a PP machine (i.e. it can have unbounded error), and so it suffices to reduce the probability of obtaining any such $y'$ being output to be strictly smaller than the probability of $y^*$ being output, in which case the PP machine is more likely to output the correct answer in its simulation of the $\PNP$ computation than not, as needed.

We develop a quantum variant of this scheme which is quite natural. However, its analysis is markedly more involved than the classical setting due to both the probabilistic nature of QMA and the possibility of ``invalid queries'' violating the QMA promise gap. To give a sense of the problems which arise, for $\PQMA$ it is no longer necessarily true that the lexicographically largest obtainable $y^*$ is a correct query string.

\subsection{Discussion and open questions}\label{sscn:discussionopenquestions}
The problems studied here explore the line of research recently initiated by Ambainis~\cite{A14} on $\PQMA$ and focus on central problems for local Hamiltonian systems. The complexity theoretic study of such problems is appealing in that it addresses the original motivation of physicist Richard Feynman in proposing quantum computers~\cite{F85}, who was interested in avenues for simulating quantum systems. Indeed, hardness results, such as Kitaev's quantum Cook-Levin theorem, rigorously justify Feynman's intuition that such simulation problems are hard. Our work (e.g. Theorem~\ref{thm:main1}) strongly supports this view by demonstrating that even some of the simplest and most natural simulation tasks, such as measuring a \emph{single qubit (!)} of a ground state, can be harder than QMA. Our study of spectral gaps (Theorem~\ref{thm:spgap}) further highlights another theme: the subtleties which must be carefully treated when studying oracle classes for promise problems (such as $\PQMA$). As quantum complexity theory commonly focuses on such {promise} problems, we believe this theme could be of interest to a broader computer science audience.

Moving to open questions, although we resolve one of the open questions from~\cite{A14}, there are others we leave open, along with some new ones. Do our results for \app\ and \tpc\ hold for more restricted classes of Hamiltonians, such as $2$-local Hamiltonians, local Hamiltonians on a 2D lattice, or specific Hamiltonian models of interest (see e.g.~\cite{CM16,PM17} for QMA-completeness results for estimating ground state energies of the spin-$1/2$ Heisenberg anti-ferromagnet)? Is \spgap\ $\PUQMA$-complete or $\PQMA$-complete (recall $\spgap\in\PQMA$ and that \cite{A14} and our work together show $\PUQMA$-hardness)? What is the relationship between $\PQMA$ and $\PUQMA$? Finally, what is the complexity of other physical tasks ``beyond'' estimating ground state energies?\\\vspace{-2mm}

\noindent\emph{Remark added later:} Since the original preprint of this paper appeared, the present authors, together with Stephen Piddock, have partially answered the first open question above by showing~\cite{GPY19} that APX-SIM remains $\PQMA$-complete for any family of local Hamiltonians which can simulate ``spatially sparse''~\cite{OT05} Hamiltonians. This, in turn, implies (e.g. using~\cite{PM17},~\cite{SV09},~\cite{PM18}) that APX-SIM remains $\PQMA$-complete even on physically motivated models, such as the Heisenberg interaction on a 2D square lattice.\\\vspace{-2mm}

\paragraph{Organization.} This paper is organized as follows: In Section~\ref{scn:preliminaries}, we give notation and formal definitions. In Section~\ref{scn:lemmas}, we show various lemmas regarding Ambainis's query Hamiltonian. In Section~\ref{scn:1local}, Section~\ref{scn:corr}, and Section~\ref{scn:spectralGap} we show Theorem~\ref{thm:main1}, Theorem~\ref{thm:main2}, and Theorem~\ref{thm:spgap}, respectively. Section~\ref{scn:PP} proves Theorem~\ref{thm:inPP}. A result used in Section~\ref{scn:lemmas} is proved in Appendix~\ref{scn:appendix}.

\section{Preliminaries}\label{scn:preliminaries}

\noindent\textbf{Notation.} For $x\in\set{0,1}^n$, $\ket{x}\in(\complex^2)^{\otimes n}$ denotes the computational basis state labeled by $x$. Let $\spa{X}$ be a complex Euclidean space. Then, $\lin{\spa{X}}$ and $\density{\spa{X}}$ denote the sets of linear and density operators acting on $\spa{X}$, respectively. For subspace $\spa{S}\subseteq\spa{X}$, $\spa{S}^\perp$ denotes the orthogonal complement of $\spa{S}$. For Hermitian operator $H$, $\lambda(H)$ and $\lambda(H|_{\spa{S}})$ denote the smallest eigenvalue of $H$ and the smallest eigenvalue of $H$ restricted to space $\spa{S}$, respectively. The spectral and trace norms are defined $\snorm{A} := \max\{\norm{A\ket{v}}_2 : \norm{\ket{v}}_2 = 1\}$ and  $\trnorm{A}:=\trace{\sqrt{A^\dagger A}}$, respectively, where $:=$ denotes a definition. We set $[m]:=\set{1,\ldots,m}$.\\\vspace{-2mm}

\noindent\textbf{Definitions and lemmas.}

\mypar{$\class{PP}$ and $\class{PQP}$} The class PP~\cite{G77} is the set of decision problems for which there exists a polynomial-time probabilistic Turing machine which accepts any YES instance with probability $>1/2$ and accepts any NO instance with probability $\leq 1/2$. The quantum analogue of PP, denoted PQP, is defined analogously to BQP except in the YES case the verifier accepts with probability $>1/2$ and in the NO case the verifier accepts with probability $\leq 1/2$.

\mypar{$\PQMA$} Defined by Ambainis~\cite{A14}, $\PQMA$ is the set of decision problems decidable by a polynomial-time deterministic Turing machine with the ability to query an oracle for a QMA-complete problem $O(\log n)$ times, where $n$ is the size of the input. For clarity,  without loss of generality we assume that for any QMA problem instance $\Pi$ the P machine wishes to solve, it applies a polynomial-time Karp/many-one reduction to map $\Pi$ to an instance of the QMA-complete problem $2$-LH~\cite{KKR06} which it then feeds to the QMA oracle. $2$-LH is defined as: Given a $2$-local Hamiltonian $H$ and inverse polynomially separated thresholds $a,b\in\reals$, decide whether $\lambda(H)\leq a$ (YES-instance) or $\lambda(H)\geq b$ (NO-instance). Note that the P machine is allowed to make queries which violate the promise gap of $2$-LH, i.e. with $\lambda(H)\in(a,b)$; in this case, the oracle can output YES or NO arbitrarily. The P machine is nevertheless required to output the same final answer (i.e. accept or reject) regardless of how such ``invalid'' queries are answered~\cite{G06}.

For any P machine $M$ making $m$ queries to a QMA oracle, we use the following terminology throughout this article. A \emph{valid} (\emph{invalid}) query satisfies (violates) the promise gap of the QMA oracle. A \emph{correct} query string $y\in\set{0,1}^m$ encodes a sequence of correct answers to all of the $m$ queries. Note that for any invalid query by $M$, any answer is considered correct, yielding the possible existence of multiple correct query strings. An \emph{incorrect} query string is one which contains at least one incorrect query answer.

\mypar{The problem APX-SIM}
\begin{definition}[$\app(H,A,k,l,a,b,\delta)$ (Ambainis~\cite{A14})]\label{dfn:apxsim}
	Given a $k$-local Hamiltonian $H$, an $l$-local observable $A$, and real numbers $a$, $b$, and $\delta$ such that $b-a\geq n^{-c}$ and $\delta\geq n^{-c'}$, for $n$ the number of qubits $H$ acts on and $c,c'>0$ some constants, decide:
	\begin{itemize}
		\item If $H$ has a ground state $\ket{\psi}$ satisfying $\bra{\psi}A\ket{\psi}\leq a$, output YES.
		\item If for all $\ket{\psi}$ satisfying $\bra{\psi}H\ket{\psi}\leq \lambda(H)+\delta$, it holds that $\bra{\psi}A\ket{\psi}\geq b$, output NO.
	\end{itemize}
\end{definition}

\mypar{Kitaev's circuit-to-Hamiltonian construction} Next, we briefly review Kitaev's circuit-to-Hamiltonian construction from the ``quantum Cook-Levin theorem''~\cite{KSV02}. Given a quantum circuit $U=U_L\cdots U_1$ consisting of $1$- and $2$-qubit gates $U_i$ and acting on registers $Q$ (proof register) and $W$ (workspace register),
this construction maps $U$ to a $5$-local Hamiltonian $H=\hin+\hout+\hprop+\hstab$ acting on registers $Q$ (proof register), $W$ (workspace register), and $C$ (clock register). Each of $H$'s terms are defined below:
\begin{eqnarray*}
    \hin&:=&
    \left(\sum_{i=1}^{p} \ketbra{1}{1}_{W_i}\right)\otimes \ketbra{0}{ 0}_C\\
    \hout&:=&\ketbra{0}{0}_{W_1}\otimes
     \ketbra{ L}{ L}_C\\
    \hprop &:=& \sum_{j=1}^{L} H_j {\rm,~where~ }H_j{\rm ~is ~defined~ as}\\
    && \hspace{-10mm}-\frac{1}{2}U_j\otimes\ketbra{ j}{ {j-1}}_C -\frac{1}{2}U_j^\dagger\otimes\ketbra{{j-1}}{ j}_C +\frac{1}{2}I\otimes(\ketbra{ j}{ j}+\ketbra{ {j-1}}{ {j-1}})_C\\
    \hstab&:=&\sum_{i=1}^{L-1}\ketbra{01}{01}_{C_i,C_{i+1}}.
\end{eqnarray*}
Above, the notation $R_i$ refers to the $i$-th qubit of any given register $R$. For any candidate proof $\ket{\psi}$ to be evaluated in expression $\trace(H\ketbra{\psi}{\psi})$, each penalty term of $H$ forces a structure on any minimizing $\ket{\psi}$: at time zero, $\hin$ ensures the workspace register $W$ is set to zero; $\hout$ checks that at the end of the computation, i.e. at time step $L$, the output qubit is close to $\ket{1}$; the propagation Hamiltonian $\hprop$ ensures all steps of $U$ appear in superposition in $\ket{\psi}$ with equal weights. Finally, although we have labeled (for ease of exposition) the clock register in $\hout$ and $\hprop$ in binary, note that in Kitaev's construction it is actually encoded in unary. In other words, time $t$ in clock register $C$ is encoded as $\ket{1^t0^{L-t}}$ (note for $\hstab$ above, register $C$ is already written in unary); $\hstab$'s role is thus to prevent invalid encodings of time steps. (For the interested reader, we remark that instead of working explicitly with Kitaev's construction, one can also consider the more general Quantum Thue System framework~\cite{BCO17}.)

In this paper, we shall use two key properties of $\hin+\hprop+\hstab$ (note $\hout$ is omitted). First, the null space of $\hin+\hprop+\hstab$ is spanned by \emph{history states}, which for any proof $\ket{\psi}$ have form
\begin{equation}\label{eqn:hist}
    \ket{\psihist}=\frac{1}{\sqrt{L+1}}\sum_{t=0}^LU_t\cdots U_1\ket{\psi}_Q\ket{0\cdots 0}_W\ket{t}_C,
\end{equation}
where $C$ is a clock register keeping track of time \cite{KSV02}, and each $U_i$ acts on at most two qubits from registers $Q$ and $W$. Second, we use the following lower bound\footnote{This bound is stated as $\Omega(\Delta/L^3)$ in \cite{GK12}; the constant $\pi^2/64$ can be derived from the analysis therein
.} on the smallest non-zero eigenvalue of $\hin+\hprop+\hstab$: \begin{lemma}[Lemma 3 (Gharibian, Kempe \cite{GK12})]\label{l:GKgap}
    The smallest non-zero eigenvalue of $\Delta(\hin+\hprop+\hstab)$ is at least $\pi^2\Delta/(64L^3)\in\Omega(\Delta/L^3)$,~for $\Delta\in\reals^+$ and $L\geq 1$.
\end{lemma}

\mypar{Miscellaneous useful facts} Let $V$ denote a QMA verification circuit acting on $M$ proof qubits with completeness $c$ and soundness $s$. If one runs $V$ on ``proof'' $\rho=I/2^M$, then for a YES instance, $V$ accepts with probability $\geq c/2^M$ (since $I/2^M$ can be viewed as ``guessing'' a correct proof with probability $\geq 1/2^M$), and in a NO instance, $V$ accepts with probability $\leq s$ (see, e.g.,~\cite{MW05,W13}).

A useful fact for complex unit vectors $\ket{v}$ and $\ket{w}$ is (see, e.g., Equation~1.33 of~\cite{G13})
\begin{equation}\label{eqn:enorm}
    \trnorm{\ketbra{v}{v}-\ketbra{w}{w}}=2\sqrt{1-\abs{\brakett{v}{w}}^2}\leq 2\enorm{\ket{v}-\ket{w}}.
\end{equation}

\section{Ambainis's Query Hamiltonian}\label{scn:lemmas}
We now show various results regarding Ambainis's ``query Hamiltonian''~\cite{A14}, which intuitively aims to have its ground space contain correct answers to a sequence of QMA queries. Let $U$ be a $\PQMA$ computation\footnote{In this work, we let $U$ denote a uniformly generated classical circuit, as opposed to a Turing machine; this is to ease the transition to quantum circuits and corresponding application of Kitaev's circuit-to-Hamiltonian construction later.}, and let $H_{\spa{Y}_i}^{i,y_{1}\cdots y_{i-1}}$ be the $2$-local Hamiltonian\footnote{While each query is encoded as a $2$-local Hamiltonian, our hardness results in Section~\ref{scn:hardness} will be for $5$- and $4$-local Hamiltonians, since the query Hamiltonians are just components of our larger, final Hamiltonian constructions.} corresponding to the $i$-th query made by $U$ given that the answers to the previous $i-1$ queries are given by $y_{1}\cdots y_{i-1}\in\set{0,1}^{0,1}$. (Without loss of generality, we may assume $H_{\spa{Y}_i}^{i,y_{1}\cdots y_{i-1}}\succeq 0$ by adding multiples of the identity and rescaling.) The oracle query made at step $i$ corresponds to an input $(H_{\spa{Y}_i}^{i,y_{1}\cdots y_{i-1}} , \epsilon, 3\epsilon)$ to $2$-LH, with (as in~\cite{A14}) completeness/``YES case'' and soundness/``NO case'' parameters $\epsilon$ and $3\epsilon$, respectively, for $\epsilon>0$ a fixed inverse polynomial. Then, Ambainis's~\cite{A14} $O(\log(n))$-local query Hamiltonian $H$ acts on $\spa{X}\otimes\spa{Y}$, where $\spa{X}=(\spa{X}_i)^{\otimes m}=(\complex^{2})^{\otimes m}$ and $\spa{Y}=\otimes_{i=1}^m\spa{Y}_i$, such that $\spa{X}_i$ is intended to encode the answer to query $i$ with $\spa{Y}_i$ encoding the ground state of the corresponding query Hamiltonian $H_{\spa{Y}_i}^{i,y_{1}\cdots y_{i-1}}$. Specifically,
	\begin{eqnarray}
		H &=& \sum_{i=1}^m\frac{1}{4^{i-1}}\sum_{y_1,\ldots,y_{i-1}}\bigotimes_{j=1}^{i-1}\ketbra{y_j}{y_j}_{\spa{X}_{j}}\otimes\left(2\epsilon \ketbra{0}{0}_{\spa{X}_{i}}\otimes I_{\spa{Y}_i} + \ketbra{1}{1}_{\spa{X}_{i}}\otimes H_{\spa{Y}_i}^{i,y_{1}\cdots y_{i-1}}\right)\nonumber\\
		&=:&\sum_{i=1}^m\frac{1}{4^{i-1}}\sum_{y_1,\ldots,y_{i-1}}M_{y_1 \cdots y_{i-1}}.\label{eqn:amb1}
	\end{eqnarray}

	Recall from Section~\ref{scn:preliminaries} that a sequence of query answers $y=y_1\cdots y_m\in\set{0,1}^m$ is \emph{correct} if it corresponds to a possible execution of $U$. Since $U$ can make queries to its QMA oracle which violate the QMA promise gap, the set of correct $y$ is generally {not} a singleton. However, we henceforth assume without loss of generality that $U$ makes at least one valid query (i.e. which satisfies the QMA promise gap). (Assume to the contrary that $U$ makes only invalid queries. Then, since recall we stipulate~\cite{G06} that the oracle's answer on any invalid query must not affect the final output of $U$, it follows that for \emph{all} possible query strings, $U$ outputs the same answer. Moreover, since $m\in O(\log n)$, we can efficiently \emph{verify} this property --- simply iterate via brute force through all possible polynomially many query strings, and check that $U$ outputs the same answer for each one. Therefore, we can assume that any reduction includes such a check as an implicit subroutine, prepended to the start of the reduction, which will handle any case in which $U$ makes only invalid queries.) We now prove the following about $H$:
\begin{lemma}\label{l:amborig}
    Define for any $x\in\set{0,1}^m$ the space
    $
        \spa{H}_{x_1\cdots x_m} := \bigotimes_{i=1}^m \ketbra{x_i}{x_i}\otimes \spa{Y}_i.
    $
    Then, there exists a correct query string $x\in\set{0,1}^m$ such that the ground state of $H$ lies in $\spa{H}_{x_1\cdots x_m}$. Moreover, suppose this space has minimum eigenvalue $\lambda$. Then, for any incorrect query string $y_1\cdots y_m$, any state in $\spa{H}_{y_1\cdots y_m}$ has energy at least $\lambda+\frac{\epsilon}{4^{m}}$.
\end{lemma}
\noindent As discussed in Section~\ref{scn:intro}, Claim 1 of~\cite{A14} proved a similar statement under the assumption that the correct query string $x$ is unique. In that setting,~\cite{A14} showed the ground state of $H$ is in $\spa{H}_{x}$, and that for \emph{all} query strings $y\neq x$, the space $\spa{H}_{y}$ has energy at least $\lambda+\frac{\epsilon}{4^{m-1}}$. However, in general invalid queries must be allowed, and in this setting this claim no longer holds --- two distinct correct query strings can have eigenvalues which are arbitrarily close if they contain queries violating the promise gap. The key observation we make here is that even in the setting of non-unique $x$, a spectral gap between the ground space and all \emph{incorrect} query strings can be shown.
\begin{proof}[Proof of Lemma \ref{l:amborig}]
    Observe first that $H$ in Equation~(\ref{eqn:amb1}) is block-diagonal with respect to register $\spa{X}$, i.e. to understand the spectrum of $H$, it suffices to understand the eigenvalues in each of the blocks corresponding to fixing $\spa{X}_i$ to some string $y\in\set{0,1}^m$. Thus, we can restrict our attention to spaces $\spa{H}_{y}$ for $y\in\set{0,1}^m$. To begin, let $x\in\set{0,1}^m$ denote a correct query string which has lowest energy among all \emph{correct} query strings against $H$, i.e. the block corresponding to $x$ has the smallest eigenvalue among such blocks. (Note that $x$ is well-defined, though it may not be unique; in this latter case, any such $x$ will suffice for our proof.) For any $y\in\set{0,1}^m$, define $\lambda_y$ as the smallest eigenvalue in block $\spa{H}_y$. We show that for any \emph{incorrect} query string $y=y_1\cdots y_m$, $\lambda_y\geq\lambda_x+\epsilon/(4^m)$.

    We use proof by contradiction, via an exchange argument\footnote{At a high level, in an \emph{exchange argument}, one begins with a solution $y$ to some problem instance, where $y$ is lacking some additional desired property $P$. One then shows how to ``exchange'' $y$ for another solution $y'$ such that $y'$ is at least as ``good'' as $y$ for the problem at hand, and so that $y'$ now has the desired property $P$. In our application here, we will exchange query string $y$ for a query string $y'$ satisfying $\lambda_{y'}<\lambda_y$. The property $P$ in this case will be whether the query string correct.}. Suppose there exists an incorrect query string $y=y_1\cdots y_m$ such that $\lambda_y <\lambda_x+\epsilon/(4^m)$. Since $y$ is an incorrect query string, there exists an $i\in[m]$ such that $y_i$ is the wrong answer to a valid query $H_{\spa{Y}_i}^{i,y_{1}\cdots y_{i-1}}$. Let $i$ denote the first such position. Now, consider operator $M_{y_1 \cdots y_{i-1}}$, which recall is defined as
    \[
       M_{y_1 \cdots y_{i-1}}= \bigotimes_{j=1}^{i-1}\ketbra{y_j}{y_j}_{\spa{X}_{j}}\otimes\left(2\epsilon \ketbra{0}{0}_{\spa{X}_{i}}\otimes I_{\spa{Y}_i} + \ketbra{1}{1}_{\spa{X}_{i}}\otimes H_{\spa{Y}_i}^{i,y_{1}\cdots y_{i-1}}\right)
    \]
    and let $\lambda_{y_1\cdots y_{i-1}\overline{y}_i}$ denote the smallest eigenvalue of $M_{y_1 \cdots y_{i-1}}$ restricted to space $\spa{H}_{y_1\cdots y_{i-1}\overline{y}_i}$, where $\overline{y}_i$ denotes the complement of bit $y_i$;
    for clarity, here we define
    \[
        \spa{H}_{y_1\cdots y_{i-1}\overline{y}_i} :=\ketbra{y_1}{y_1}_{\spa{X}_1}\otimes\cdots\otimes \ketbra{\overline{y}_i}{\overline{y}_i}_{\spa{X}_i}\otimes\left(\bigotimes_{j=i+1}^{m}\spa{X}_j\right)\otimes \spa{Y}.
    \]
    Note that string $y_1\cdots y_{i-1}\overline{y}_i$ corresponds to $i$ correct query bits, with $\overline{y}_i$ the correct answer to query $i$.
    Then, we claim that any state $\ket{\phi}\in \spa{H}_{y_1\cdots y_i}$ (i.e. $\ket{\phi}$ is of the same dimension as $H$, with registers $\spa{X}_1$ through $\spa{X}_i$ fixed to state $\ket{y_1}\otimes\cdots \otimes \ket{y_i}$) satisfies the bound of the following lemma.
    \begin{lemma}\label{l:bound1}
        For $\ket{\phi}$, $M_{y_1\cdots y_{i-1}}$, and $\lambda_{y_1\cdots y_{i-1}\overline{y}_i}$ as defined above, we have
        \[
            \bra{\phi}M_{y_1\cdots y_{i-1}}\ket{\phi}\geq \lambda_{y_1\cdots y_{i-1}\overline{y}_i}+\epsilon.
        \]
    \end{lemma}
    \noindent Here, for clarity, $M_{y_1\cdots y_{i-1}}$ is of the same dimension as $\ket{\phi}$, since we assume $M_{y_1\cdots y_{i-1}}$ has an implicit tensor product with term $I_{\spa{X}_{i+1}\otimes\spa{Y}_{i+1}\cdots\otimes\spa{X}_{m}\otimes\spa{Y}_{m}}$. (This is intuitively because $M_{y_1\cdots y_{i-1}}$ corresponds to making query $i$ given responses to queries $1$ through $i-1$, and queries $i+1$ through $m$ are yet to be determined.)
    \begin{proof}
        Constrained to space $\spa{H}_{y_1\cdots y_{i-1}}$, $M_{y_1\cdots y_{i-1}}$ reduces to operator $M':=2\epsilon \ketbra{0}{0}_{\spa{X}_{i}}\otimes I_{\spa{Y}_i} + \ketbra{1}{1}_{\spa{X}_{i}}\otimes H_{\spa{Y}_i}^{i,y_{1}\cdots y_{i-1}}$, which is block-diagonal in the standard basis with respect to $\spa{X}_i$. Thus, if query $i$ is a YES instance, the smallest eigenvalue of $M'$ lies in the block corresponding to setting $\spa{X}_i$ to (the correct query answer) $\ket{1}$, and is at most $\epsilon$ (since we are assuming each query has YES and NO energy thresholds $\epsilon$ and $3\epsilon$). On the other hand, the block with $\spa{X}_i$ set to $\ket{0}$ by definition reduces to $2\epsilon I_{\spa{Y}_i}$, and hence has all eigenvalues equaling $2\epsilon$. Thus, setting $\spa{X}_i$ to the correct query answer $\ket{1}$ yields an energy penalty which is at least $\epsilon$ less than the penalty obtained for setting $\spa{X}_i$ to the wrong query answer, $\ket{0}$.  Analogously, when query $i$ is a NO instance (i.e. the correct query answer is $\ket{0}$), a similar argument shows the $\ket{0}$-block again has eigenvalues equaling $2\epsilon$ and now the $\ket{1}$-block has eigenvalues at least $3\epsilon$. We conclude that in both YES and NO cases, setting $\spa{X}_i$ to the correct query answer achieves an energy penalty at least $\epsilon$ less than setting $\spa{X}_i$ to the wrong query answer. As a result,
        \[
            \lambda_{y_1\cdots y_{i-1}\overline{y}_i}\leq \lambda_{y_1\cdots {y}_i}-\epsilon,
        \]
        which implies the claim.
    \end{proof}

    With Lemma~\ref{l:bound1} in hand, we return to our exchange argument. Let $\widehat{M}_{y_1\cdots y_t}$ denote the set of terms from Equation~(\ref{eqn:amb1}) which are consistent with some prefix $y_1\ldots y_t$ (e.g. $M_{y_1\ldots y_t}$, $M_{y_1\ldots y_t0}$, $M_{y_1\ldots y_t 1}$, etc).
    Recall that the bits $y_{i+1}\cdots y_{m}$ form the tail of the string $y$ following the incorrect query answer $y_i$. Set $y'_{i+1}\cdots y'_{m}$ so that $y':=y_1\cdots \overline{y}_i y'_{i+1}\cdots y'_m$ is a correct query string.
    Care is now required; the new query bits $y'_{i+1}\cdots y'_m$ may lead to different energy penalties than the previous string $y_{i+1}\cdots y_m$ against the Hamiltonian terms in set $\widehat{M}_{y_1\cdots \overline{y}_i}$. In other words, we must upper bound any possible energy penalty \emph{increase} when exchanging $y_1\cdots \overline{y}_i y_{i+1}\cdots y_m$ for $y'$. To do so, recall that all Hamiltonian terms in Equation~(\ref{eqn:amb1}) are positive semidefinite. Thus, for any state $\ket{\psi}$ in space $\spa{H}_{y_1\cdots \overline{y}_i}$, the energy obtained by $\ket{\psi}$ against terms in $\widehat{M}_{y_1\cdots \overline{y}_i}$ is at least $0$. Conversely, in the worst case, since each term in $\widehat{M}_{y_1\cdots \overline{y}_i}$ has minimum eigenvalue at most $2\epsilon$, the eigenvector $\ket{\psi}$ of smallest eigenvalue in block $H_{y'}$ incurs an additional penalty against $H$ for queries $i+1$ through $m$ of at most (taking into account the normalization factor $1/4^{i-1}$ in Equation~(\ref{eqn:amb1}))
    \[
        2\epsilon\sum_{k=i}^\infty \frac{1}{4^k}=\frac{2\epsilon}{3\cdot4^{i-1}}.
    \]
     We conclude that (again taking into account the normalization factor $1/4^{i-1}$ in Equation~(\ref{eqn:amb1}))
     \[
         \lambda_{y'}\leq \lambda_y-\frac{\epsilon}{4^{i-1}}+\frac{2\epsilon}{3\cdot4^{i-1}} < \left(\lambda_x+\frac{\epsilon}{4^m}\right)-\frac{\epsilon}{4^{i-1}}+\frac{2\epsilon}{3\cdot4^{i-1}}<\lambda_x
     \]
     where the second inequality follows by the assumption $\lambda_y<\lambda_x+{\epsilon}/{4^m}$. This is a contradiction.
\end{proof}

The next lemma converts $H$ from an $O(\log n)$-local Hamiltonian to an $O(1)$-local one.
\begin{lemma}\label{l:amb}
	For any $x\in\set{0,1}^m$, let $\hat{x}$ denote its unary encoding. Then, for any $\PQMA$ circuit $U$ acting on $n$ bits and making $m\geq 1$ queries to a QMA oracle, there exists a mapping to a $4$-local Hamiltonian $H'$ acting on space $(\complex^2)^{\otimes 2^m-1}\otimes\spa{Y}$ such that there exists a correct query string $x=x_1\cdots x_m$ satisfying:
	\begin{enumerate}
		\item The ground state of $H'$ lies in subspace $\ketbra{\hat{x}}{\hat{x}}\otimes \spa{Y}$.
		\item For any state $\ket{\psi}$ in subspace $\ketbra{\hat{x}'}{\hat{x}'}\otimes \spa{Y}$ where either $\hat{x}'$ is not a unary encoding of a binary string $x'$ or $x'$ is an incorrect query string, one has $\bra{\psi}H'\ket{\psi}\geq \lambda(H')+\epsilon/4^{m}$, for  inverse polynomial $\epsilon$.
		\item For all strings $x'\in\set{0,1}^m$, $H'$ acts invariantly on subspace $\ketbra{\hat{x}'}{\hat{x}'}\otimes \spa{Y}$.
\item The mapping can be computed in time polynomial in $n$ (recall $m\in O(\log n)$).
	\end{enumerate}
\end{lemma}
\begin{proof}
	We show how to improve the $O(\log n)$-local construction $H$ of Lemma~\ref{l:amborig} to $4$-local $H'$. Specifically, recall that $H$ from Equation~(\ref{eqn:amb1}) acts on $(\spa{X}\otimes\spa{Y})$. Using a trick of Kitaev~\cite{KSV02}, we encode the $\spa{X}=\spa{X}_1\otimes\cdots\otimes \spa{X}_m$ register in unary. Specifically, we can write
	\begin{eqnarray*}
		M_{y_1\cdots y_{i-1}}&=&\sum_{y_{i+1},\ldots,y_{m}}2\epsilon\bigotimes_{j=1}^{i-1}\ketbra{y_j}{y_j}_{\spa{X}_{j}}\otimes\ketbra{0}{0}_{\spa{X}_{i}}\bigotimes_{k=i+1}^{m}\ketbra{y_k}{y_k}_{\spa{X}_{k}}\otimes I_{\spa{Y}} +\\		&&\mbox{\hspace{3mm}}\bigotimes_{j=1}^{i-1}\ketbra{y_j}{y_j}_{\spa{X}_{j}}\otimes\ketbra{1}{1}_{\spa{X}_{i}}\bigotimes_{k=i+1}^{m}\ketbra{y_k}{y_k}_{\spa{X}_{k}}\otimes H_{\spa{Y}_i}^{i,y_{1}\cdots y_{i-1}}.
	\end{eqnarray*}
	We now replace register $\spa{X}_1\otimes\cdots\otimes \spa{X}_m$ with register $\spa{X}'=(\complex^2)^{\otimes 2^m-1}$ and encode each binary string $x\in\set{0,1}^m$ as the unary string $\hat{x}=\ket{1}^{\otimes \abs{x}}\ket{0}^{\otimes 2^m-\abs{x}-1}$, where $\abs{x}$ is the non-negative integer corresponding to string $x$. In other words, for $M_{x_1\cdots x_{i-1}}$, we replace each string $\ketbra{x}{x}_{\spa{X}_1\otimes\cdots\otimes\spa{X}_m}$ with $\ketbra{1}{1}_{\spa{X}_{1}\otimes\cdots\otimes\spa{X}_{\abs{x}}}\otimes\ketbra{0}{0}_{\spa{X}_{\abs{x}+1}\otimes\cdots\otimes\spa{X}_{2^m-1}}$. Denote the resulting Hamiltonian as $H_1$.
	
	To ensure states in $\spa{X}'$ follow this encoding, add a weighted version of Kitaev's~\cite{KSV02} penalty Hamiltonian, introduced in Section~\ref{scn:preliminaries},
	\[
		\hstab=3\epsilon\sum_{j=1}^{2^m-2}\ketbra{0}{0}_j\otimes\ketbra{1}{1}_{j+1},
	\]
	 i.e., our final Hamiltonian is $H'=H_1+\hstab$. To show that $H'$ satisfies the same properties as $H$ as stated in the claim, we follow the analysis of Kitaev~\cite{KSV02}. Namely, partition the space $\spa{X}'\otimes\spa{Y}$ into orthogonal spaces $\spa{S}$ and $\spa{S}^\perp$ corresponding to the space of valid and invalid unary encodings of $\spa{X}'$, respectively. Since $H_1$ and $\hstab$ act invariantly on $\spa{S}$ and $\spa{S}^\perp$, we can consider $\spa{S}$ and $\spa{S}^\perp$ separately. In $\spa{S}$, $H'$ is identical to $H$, implying the claim. In $\spa{S}^\perp$, the smallest non-zero eigenvalue of $\hstab$ is at least $3\epsilon$. Thus, since $H_1\succeq 0$, if we can show that the smallest eigenvalue of $H$ is at most $3\epsilon-\epsilon/4^m$, we have shown the claim (since, in particular, we will have satisfied statement 2 of our claim). To show this bound on the smallest eigenvalue, suppose $x$ is all zeros, i.e. set register $\spa{X}_1\otimes\cdots\otimes\spa{X}_m$ for $H$ to all zeros. Then, each term $M_{0_1\cdots 0_{i-1}}$ yields an energy penalty of exactly $2\epsilon$, yielding an upper bound on the smallest eigenvalue of $H$ of
$
    2\epsilon\sum_{k=0}^{m-1}\frac{1}{4^k}\leq \frac{8}{3}\epsilon=3\epsilon-\epsilon/3.
$
\end{proof}

\section{Lower bounds (hardness results)}\label{scn:hardness}
\subsection{$\PQMA$-completeness of \app}\label{scn:1local}

In this section, we restate and prove Theorem~\ref{thm:main1}. For this, we first show an extension of the Projection Lemma of Kempe, Kitaev, Regev~\cite{KKR06}.

\begin{lemma}[Extended Projection Lemma (\emph{cf.~\cite{KKR06}})]\label{cor:kkr}
	Let $H=H_1+H_2$ be the sum of two Hamiltonians operating on some Hilbert space $\spa{H}=\spa{S}+\spa{S}^\perp$. The Hamiltonian $H_1$ is such that $\spa{S}$ is a zero eigenspace and the eigenvectors in $\spa{S}^\perp$ have eigenvalue at least $J>2\snorm{H_2}$. Let $K:=\snorm{H_2}$. Then, for any $\delta\geq0$ and $\ket{\psi}$ satisfying $\bra{\psi}H\ket{\psi}\leq \lambda(H)+\delta$, there exists a $\ket{\psi'}\in \spa{S}$ such that:
	\begin{itemize}
        \item (Ground state energy bound)
        \[
		\lambda(H_2|_{\spa{S}})-\frac{K^2}{J-2K}\leq \lambda(H)\leq \lambda(H_2|_{\spa{S}}),
	\]
where $\lambda(H_2|_{\spa{S}})$ denotes the smallest eigenvalue of $H_2$ restricted to space $\spa{S}$.
        \item (Ground state deviation bound)
        \[
            \abs{\brakett{\psi}{\psi'}}^2\geq {1-\left(\frac{K+\sqrt{K^2+\delta(J-2K)}}{J-2K}\right)^2}.
            \]
        \item  (Energy obtained by perturbed state against $H$)
        \[
            \bra{\psi'}H\ket{\psi'}\leq\lambda(H)+\delta+2K\frac{K+\sqrt{K^2+\delta(J-2K)}}{J-2K}.
        \]
    \end{itemize}
\end{lemma}
\begin{proof}
 The first claim is the original Projection Lemma~\cite{KKR06} (which we do not require in this work, but which we also state above for completeness). The second and third claims we prove here. Consider arbitrary $\ket{\psi}$ such that $\bra{\psi}H\ket{\psi}\leq \lambda(H)+\delta$. We can write $\ket{\psi}=\alpha_1\ket{\psi_1}+\alpha_2\ket{\psi_2}$ for unit vectors $\ket{\psi_1}\in \spa{S}$, $\ket{\psi_2}\in \spa{S}^\perp$, $\alpha_1,\alpha_2\in\reals$, $\alpha_1,\alpha_2\geq 0$, and $\alpha_1^2+\alpha_2^2=1$. Following the proof of the Projection Lemma of ~\cite{KKR06}, we have
	\begin{eqnarray}
				\bra{\psi}H\ket{\psi}&\geq& \bra{\psi}H_2\ket{\psi}+J\alpha_2^2\nonumber\\
				&=&(1-\alpha_2^2)\bra{\psi_1}H_2\ket{\psi_1}+2\alpha_1\alpha_2\operatorname{Re}\bra{\psi_1}H_2\ket{\psi_2}+\nonumber\\
&&\alpha_2^2\bra{\psi_2}H_2\ket{\psi_2}+J\alpha_2^2\nonumber\\
				&\geq&\bra{\psi_1}H_2\ket{\psi_1}-K(\alpha_2^2+2\alpha_2+\alpha_2^2)+J\alpha_2^2\nonumber\\
				&=&\bra{\psi_1}H\ket{\psi_1}+(J-2K)\alpha_2^2-2K\alpha_2\label{eqn:key1},
	\end{eqnarray}
    where the last equality follows since $\ket{\psi_1}\in \spa{S}$. Since by assumption $\bra{\psi}H\ket{\psi}\leq \lambda(H)+\delta$, Equation~(\ref{eqn:key1}) implies $\lambda(H)+\delta\geq \bra{\psi_1}H\ket{\psi_1}+ (J-2K)\alpha_2^2-2K\alpha_2$. Combined with the fact $\lambda(H)\leq\bra{\psi_1}H\ket{\psi_1}$, we have
	\[
		0\geq\lambda(H)- \bra{\psi_1}H\ket{\psi_1}\geq (J-2K)\alpha_2^2-2K\alpha_2-\delta,
	\]
	which holds only if
	$
		\abs{\alpha_2}\leq \frac{K+\sqrt{K^2+\delta(J-2K)}}{J-2K}.
	$
	Thus, setting $\ket{\psi'}:=\ket{\psi_1}$ yields the first claim. To see that $\ket{\psi'}$ also satisfies the second claim, by Equation~(\ref{eqn:key1}) we have
	\[
		\lambda(H)+\delta\geq \bra{\psi}H\ket{\psi}\geq  \bra{\psi_1}H\ket{\psi_1} + (J-2K)\alpha_2^2-2K\alpha_2.
	\]
The assumption $J>2K$ implies $(J-2K)\alpha_2^2\geq 0$, and so
    \[
        \bra{\psi_1}H\ket{\psi_1} \leq  \lambda(H)+\delta + 2K\abs{\alpha_2}.
    \]
    Plugging in our upper bound on $\abs{\alpha_2}$ now yields the second claim.
\end{proof}

With Lemma~\ref{cor:kkr} in hand, we now prove the main result of this section.

\begin{reptheorem}{thm:main1}
	$\app$~is $\PQMA$-complete for $k=5$ and $l=1$, i.e., for $5$-local Hamiltonian $H$ and $1$-local observable $A$.
\end{reptheorem}

\begin{proof}
Containment in $\PQMA$ was shown\footnote{Intuitively, the idea is as follows. First, the P machine uses $O(\log n)$ queries to perform a version of binary search tailored to promise problems to estimate the ground state energy of the given Hamiltonian $H$ to within additive inverse polynomial error; call this estimate $g$. Then, the P machine makes one final QMA query: Does $H$ have a ground state $\ket{\psi}$ with ground state energy approximately $g$ and with small expectation against the given observable $A$?} for $k,l\in O(\log n)$ in \cite{A14}, for $n$ the input size. We show $\PQMA$-hardness of $\app$ via a polynomial-time mapping/Karp reduction.

Let $U'$ be an arbitrary $\PQMA$ circuit\footnote{As noted in Section~\ref{scn:lemmas}, the base class P is normally defined using a polynomial-time deterministic Turing machine. In this paper, we assume the action of the P machine is, without loss of generality, instead expressed via a polynomial-time uniform family of quantum circuits.} for instance $\Pi$, such that $U'$ acts on workspace register $W$ and query result register $Q$. Suppose $U'$ consists of $L'$ gates and makes $m=c\log(n)$ queries, for $c\in O(1)$. Without loss of generality, $U'$ can be simulated with a similar unitary $U$ which treats $Q$ as a \emph{proof} register which it does not alter at any point: Namely, $U$ does not have access to a $\QMA$ oracle, but rather reads bit $Q_i$ whenever it desires the answer to the $i$-th query. Thus, if a correct query string $y_1\cdots y_m$ corresponding to an execution of $U'$ on input $x$ is provided in $Q$ as a ``proof'', then the output statistics of $U'$ and $U$ are identical. We can also assume that $Q$ is encoded not in binary, but in unary. Thus, $Q$ consists of $2^m-1\in\poly(n)$ bits. For simplicity, however, in our discussion we will speak of $m$-bit query strings $y=y_1\cdots y_m$ in register $Q$.

\paragraph{The construction.} We now construct our instance of $\app$. First, we map $U$ to a $5$-local Hamiltonian $H_1$ via a modification of the circuit-to-Hamiltonian construction of Kitaev~\cite{KSV02}, such that $H_1$ acts on registers $W$ (workspace register), $Q$ (proof register), and $C$ (clock register). Recall (Section~\ref{scn:preliminaries}) that Kitaev's construction outputs Hamiltonian terms $\hin+\hprop+\hstab+\hout$. Set $H_1=\Delta(\hin+\hprop+\hstab)$ for $\Delta$ to be set as needed. In contrast to~\cite{KSV02}, it is crucial that $\hout$ be omitted from $H_1$, as we require our final Hamiltonian $H$ to enforce a certain structure on the ground space  \emph{regardless} of whether the computation should accept or reject. The job of ``checking the output'' will instead be delegated to the observable $A$. Formally, following the discussion in Section~\ref{scn:preliminaries} regarding Equation~(\ref{eqn:hist}), $H_1$ has a non-trivial null space, which is its ground space, consisting of history states $\ket{\psihist}$ (see Equation~(\ref{eqn:hist})) which simulate $U$ on registers $W$ and $Q$. These history states correctly simulate $U'$ \emph{assuming that} $Q$ is initialized to a correct proof.

	To thus enforce that $Q$ is initialized to a correct proof, let $H_2$ be our variant $H'$ of Ambainis's query Hamiltonian from Lemma~\ref{l:amb}, such that $H_2$ acts on registers $Q$ and $Q'$, where for clarity $Q=(\complex^2)^{\otimes 2^m-1}$ for $m\in O(\log n)$; moreover, setting $Q'=\spa{Y}$ from Lemma~\ref{l:amb} enforces the ground space of $H_2$ to contain unary encodings of valid strings of query answers $y_1\cdots y_m$ for input $\Pi$ in register $Q$, as desired. Our final Hamiltonian is $H=H_1+H_2$, which is $5$-local since $H_1$ is $5$-local and $H_2$ is $4$-local.

Suppose without loss of generality that $U$'s output qubit is $W_1$, which is set to $\ket{0}$ until the final time step, in which the correct output is copied to it. Then, set observable $A=(I+Z)/2$ such that $A$ acts on qubit $W_1$. Set $a=1-1/(L+1)$, and $b=1-1/2L$ for $L$ the number of gates in $U$. Fix $\eta\geq\max(\snorm{H_2},1)$ (such an $\eta$ can be efficiently computed by applying the triangle inequality and summing the spectral norms of each term of $H_2$ individually). Set $\Delta= 128L^3\eta\gamma$ for $\gamma$ a monotonically increasing polynomial function of $L$ to be set as needed. Finally, set $\delta=1/\Delta$. This completes the construction.
\\\vspace{-2mm}
	
\noindent\textbf{Correctness.} We now prove correctness, meaning that if $\Pi$ is a YES (NO) instance of a $\PQMA$ problem, then $\app(H, A, 5, 1, a, b, \delta)$ is a YES (NO) instance of $\app$. The following pair of lemmas together show this. They both assume the terminology set up in this proof thus far.

\begin{lemma}
    Suppose $\Pi$ is a YES instance of a $\PQMA$ problem. Then, $(H,A,k,l,a,b,\delta)$ is a YES instance of $\app$.
\end{lemma}
\begin{proof}
    By Lemma~\ref{l:amb}, the ground space of $H_2$ is spanned by states of the form $\ket{\hat{x}}_Q\otimes\ket{\phi}_{Q'}$ where $\hat{x}$ is a correct query string encoded in unary. Fix an arbitrary such ground state $\ket{\hat{x}}_Q\otimes\ket{\phi}_{Q'}$. Since $\Pi$ is a YES instance, setting $Q$ to $\hat{x}$ in this manner causes $U=U_L\cdots U_1$ to accept with certainty. Consider a version of the history state $\ket{\psihist}$ from Equation~\eqref{eqn:hist} on registers $W$ (workspace), $C$ (clock), $Q$, and $Q'$ ($Q$ and $Q'$ together are the ``proof register'', where recall $U$ reads query answers from $Q$ but does not access $Q'$ at any point),
\[
    \ket{\psihist}=\frac{1}{\sqrt{L+1}}\sum_{t=0}^LU_t\cdots U_1\ket{\hat{x}}_Q\ket{\phi}_{Q'}\ket{0\cdots 0}_W\ket{\hat{t}}_C,
\]
which by construction lies in the ground space of $H_1$. Since $U$ can read but does not alter the contents of $Q$, the history state has the tensor product form
\[
     \ket{\psihist}=\ket{\psihist'(x)}_{W,C}\otimes\ket{\hat{x}}_Q\otimes \ket{\phi}_{Q'}
\]
for appropriately defined $\ket{\psihist'(x)}_{W,C}$. Since the state encodes a correct query string in register $Q$, $\ket{\psihist}$ lies in the ground space of $H_2$. We conclude that $\ket{\psihist}$ is in the joint ground space of $H_1$ and $H_2$, and hence in the ground space of $H=H_1+H_2$. Finally, by assumption, the output qubit $W_1$ of the construction is set to $\ket{0}$ in timesteps 0 to $L-1$. In timestep $L$, since $U$ accepts with certainty given the correct query string $\hat{x}$,  $W_1$ is set to $\ket{1}$. It follows that, as desired,
\[
    \bra{\psihist}A\ket{\psihist}=\frac{1}{2}\bra{\psihist}(I+Z)\ket{\psihist}=1-\frac{1}{L+1}=a.
\]
\end{proof}	

\begin{lemma}\label{l:no}
    Suppose $\Pi$ is a NO instance of a $\PQMA$ problem. Then, $(H,A,k,l,a,b,\delta)$ is a NO instance of $\app$.
\end{lemma}
\begin{proof}
	Consider any low energy state $\ket{\psi}$ satisfying $\bra{\psi}H\ket{\psi}\leq \lambda(H)+\delta$. By Lemma~\ref{l:GKgap}, the smallest non-zero eigenvalue of $H_1$ is at least $J=\pi^2\Delta/(64L^3)= \pi^2 \eta\gamma/64$. Recalling that $\delta=1/\Delta$, apply Lemma~\ref{cor:kkr} to obtain that there exists a valid history state $\ket{\psi'}$ on $W$, $C$, $Q$, and $Q'$ such that $\abs{\brakett{\psi}{\psi'}}^2\geq 1-O(\gamma^{-2})$ (see Lemma~\ref{l:proof} in Appendix~\ref{scn:appendix} for a proof), which by Equation~(\ref{eqn:enorm}) implies
\begin{equation}\label{eqn:2}
    \trnorm{\ketbra{\psi}{\psi}-\ketbra{\psi'}{\psi'}}\leq\frac{c}{\gamma}
\end{equation}
for some constant $c>0$. By definition, such a history state $\ket{\psi'}$ simulates $U$ given ``quantum proof'' $\ket{\phi}_{Q,Q'}$ in registers $Q$ and $Q'$, i.e.
\[
    \ket{\psi'}=\sum_{t} U_t\cdots U_1 \ket{\phi}_{Q,Q'}\ket{0\cdots 0}_W\ket{t}_C.
\]
Lemma~\ref{cor:kkr} and the proof of Lemma~\ref{l:proof} also give that
	\[
	\bra{\psi'}H\ket{\psi'} \leq \lambda(H) + \delta + ~O\left(\frac{\eta}{\gamma^2}\right) =: \lambda(H) + \delta + \gamma' .
	\]

Unfortunately, \emph{a priori} the state $\ket{\phi}$ in the proof register $(Q,Q')$ of $\ket{\psi'}$ is arbitrary. Let us now approximate $\ket{\psi'}$ with another history state $\ket{\psi''}$, which contains a string of \emph{correct} query answers in $Q$. This is accomplished by the following lemma, which assumes the definitions introduced in this proof thus far.

\begin{lemma}\label{l:other}
    There exists a history state $\ket{\psi''}$ such that
    \begin{equation}\label{eqn:4}
        \trnorm{\ketbra{\psi}{\psi}-\ketbra{\psi''}{\psi''}}\leq\frac{c}{\gamma}+2\sqrt{\frac{4^m(\delta+\gamma')}{\epsilon}}
    \end{equation}
    and where register $Q$ of $\ket{\psi''}$ contains only correct query answers when written in the standard basis.
\end{lemma}
\begin{proof}
By Lemma~\ref{l:amb}, the ground space $\spa{G}$ of $H_2$ is contained in the span of states of the form $\ket{\hat{x}}_Q\otimes\ket{\phi'}_{Q'}$ where $\hat{x}$ is a correct query string encoded in unary. Since the ground spaces of $H_1$ and $H_2$ have non-empty intersection, i.e. history states acting on ``quantum proofs'' from $\spa{G}$ (which lie in the null space of $H_1$ and obtain energy $\lambda(H_2)$ against $H_2$), we know $\lambda(H)=\lambda(H_2)$. Thus, since $H_1\succeq 0$,
\begin{equation}\label{eqn:3}
    \bra{\psi'}H_2\ket{\psi'}\leq \bra{\psi'}H\ket{\psi'}\leq\lambda(H_2)+(\delta+\gamma').
\end{equation}
 Write $\ket{\phi}=\alpha\ket{\phi_1}+\beta\ket{\phi_2}$ for
    $\ket{\phi_1}\in\operatorname{Span}\set{\ket{\hat{x}}_Q\otimes\ket{\phi'}_{Q'}\mid \text{correct query string } x}$ and $\ket{\phi_2}\in\operatorname{Span}\set{\ket{\hat{x}}_Q\otimes\ket{\phi'}_{Q'}\mid \text{incorrect query string } x}$ ($\ket{\phi_1}$, $\ket{\phi_2}$ normalized), $\alpha ,\beta\in \complex , \abs{\alpha}^2+\abs{\beta}^2=1$.
Since any history state $\ket{\psi'}$, for any amplitudes $\alpha_{x}$ and unit vectors $\ket{\phi'_x}$, has form
\[
    \sum_{t,x}\alpha_{x}U_t\cdots U_1 \ket{0\cdots 0}_W\ket{t}_C\ket{\hat{x}}_{Q}\ket{\phi'_x}_{Q'}=\sum_{x}\alpha_{x}\ket{\psihist'(x)}_{W,C}\ket{\hat{x}}_{Q}\ket{\phi'_x}_{Q'}
\]
 (i.e. for any fixed $x$, $\ket{\hat{x}}_Q$ is not altered), and since $H_2$ is block-diagonal with respect to the standard basis in $Q$, by Equation~(\ref{eqn:3}) and Lemma~\ref{l:amb} we have
\begin{eqnarray*}
    \lambda(H_2)+(\delta+\gamma')&\geq& \bra{\psi'}H_2\ket{\psi'}= \abs{\alpha}^2\bra{\phi_1}H_2\ket{\phi_1}+\abs{\beta}^2\bra{\phi_2}H_2\ket{\phi_2}\\
    &\geq&\abs{\alpha}^2\lambda(H_2)+\abs{\beta}^2\left(\lambda(H_2)+\frac{\epsilon}{4^m}\right),
\end{eqnarray*}
which implies $\abs{\beta}^2\leq 4^m(\delta+\gamma')/\epsilon$. Thus, defining $\ket{\psi''}$ as the history state for ``proof'' $\ket{\phi_1}_{Q,Q'}$, we have that $\trnorm{\ketbra{\psi}{\psi}-\ketbra{\psi''}{\psi''}}$ is at most
\begin{equation*}
    \trnorm{\ketbra{\psi}{\psi}-\ketbra{\psi'}{\psi'}}+\trnorm{\ketbra{\phi}{\phi}-\ketbra{\phi_1}{\phi_1}}
    \leq\frac{c}{\gamma}+2\sqrt{\frac{4^m(\delta+\gamma')}{\epsilon}},
\end{equation*}
which follows from the triangle inequality, Equation~(\ref{eqn:2}), and the structure of the history state.
\end{proof}

With Lemma~\ref{l:other} in hand, we are ready to complete the proof of Lemma~\ref{l:no}. Observe that increasing $\gamma$ by a polynomial factor decreases $\gamma'$ by a polynomial factor,\footnote{In a previous version of this article, $\gamma'\in\Theta(\snorm{H}/\gamma)$. But this was circularly defined (recall $H=H_1+H_2$), since $\snorm{H_1}$ scales with $\Delta$, which in turn scales with $\gamma$. In the present version of the article, $\gamma'\in O\left(\eta/\gamma^2\right)$, where $\eta$ depends only on $\snorm{H_2}$, and is thus independent of $\gamma$. Hence, increasing $\gamma$ now correctly decreases $\gamma'$.} so $\delta+\gamma'$ in Equation~(\ref{eqn:4}) also decreases by a polynomial factor. Thus, set $\gamma$ as a large enough polynomial in $L$ such that
\begin{equation}\label{eqn:5}
    \frac{c}{\gamma}+2\sqrt{\frac{4^m(\delta+\gamma')}{\epsilon}}\leq \frac{1}{2L}.
\end{equation}
 Since $U$ rejects any correct query string (with certainty) in the NO case, and since $\ket{\psi''}$ is a valid history state whose $Q$ register is a superposition over correct query strings (all of which must lead to reject), we conclude that $\bra{\psi''}A\ket{\psi''}=1$. Moreover, we have that
 \[
 \abs{\trace(A\ketbra{\psi}{\psi})-\trace(A\ketbra{\psi''}{\psi''})}\leq
		\snorm{A}\trnorm{\ketbra{\psi}{\psi}-\ketbra{\psi''}{\psi''}}\leq \frac{1}{2L},
\]
	where the first inequality follows from H\"{o}lder's inequality, and the second by Equations~(\ref{eqn:4}) and~(\ref{eqn:5}). We conclude that $\bra{\psi}A\ket{\psi}\geq 1-1/(2L)$, completing the proof.
\end{proof}
\end{proof}

\subsection{$\PQMA$-completeness of \tpc}\label{scn:corr}

We now define \tpc\ and show that it is $\PQMA$-complete using similar techniques to Section~\ref{scn:1local}. For brevity, define $f(\ket{\psi},A,B):= \bra{\psi}A\otimes B \ket{\psi}-\bra{\psi}A\ket{\psi}\bra{\psi}B\ket{\psi}$.

\begin{definition}[\tpc$(H,A,B,k,l,a,b,\delta)$]
	Given a $k$-local Hamiltonian $H$, $l$-local observables $A$ and $B$, and real numbers $a$, $b$, and $\delta$ such that $a-b\geq n^{-c}$ and $\delta\geq n^{-c'}$, for $n$ the number of qubits $H$ acts on and $c,c'\geq 0$ some constants, decide:

	\begin{itemize}
		\item If $H$ has a ground state $\ket{\psi}$ satisfying $f(\ket{\psi},A, B)\geq a$, output YES.
		\item If for any $\ket{\psi}$ satisfying $\bra{\psi}H\ket{\psi}\leq \lambda (H)+\delta$ it holds that $f(\ket{\psi}, A, B)\leq b$, output NO.
	\end{itemize}
\end{definition}

\noindent We now prove Theorem~\ref{thm:main2} by showing $\PQMA$-hardness in Lemma~\ref{lem:2-PHard} and containment in $\PQMA$ in Lemma~\ref{lem:2-PIn}.

\begin{lemma} \label{lem:2-PHard}
    \tpc\ is $\PQMA$-hard for $k=5$ and $l=1$, i.e., for $5$-local Hamiltonian $H$ and $1$-local observables $A$ and $B$.
\end{lemma}

\begin{proof}
	For an arbitrary $\PQMA$ circuit $U'$, define $U$ as in the proof of Theorem \ref{thm:main1}, consisting of $L$ one- and two-qubit gates. We modify $U$ as follows. Let $U$'s output qubit be denoted $W_1$. We add two ancilla qubits, $W_2$ and $W_3$, which are set to $\ket{00}$ throughout $U$'s computation. We then append to $U$ a sequence of six 2-qubit gates which, controlled on $W_1$, map $\ket{00}$ in $W_2W_3$ to $\ket{\phi^+}=(\ket{00}+\ket{11})/\sqrt{2}$, e.g. apply a controlled Hadamard gate and the 5-gate Toffoli construction from Figure~4.8 of~\cite{NC00}. Appending six identity gates on $W_1$, we obtain a circuit $V=V_{L+12}\cdots V_1$ which has $L+12$ gates. Finally, we construct $H=H_1+H_2$ as in the proof of Theorem \ref{thm:main1}, mapping $V$ to a 5-local Hamiltonian $H_1$ on registers $W$, $Q$, and $C$, and we set $A={Z}_{W_2}$ and $B={Z}_{W_3}$ for Pauli $Z$. Similar to the proof of Theorem~\ref{thm:main1}, set $\Delta= 128L^3\eta\gamma$ and $\delta=1/\Delta$, for $\gamma$ large enough so that
\begin{equation}\label{eqn:6}
    \frac{c}{\gamma}+2\sqrt{\frac{4^m(\delta+\gamma')}{\epsilon}}\leq \frac{1}{2(L+13)},
\end{equation}
for $\gamma '$ as defined in the proof of Theorem~\ref{thm:main1}.	Set $a=3/(L+13)$ and $b=1/(L+13)$. This completes the construction.

	To set up the correctness proof, consider history state $\ket{\psihist}$ for $V$ given quantum proof $\ket{\phi}_{Q,Q'}$, and define for brevity $\ket{\phi_t}:=V_t\cdots V_1\ket{\phi}_{Q,Q'}\ket{0\cdots 0}_W\ket{00}_{W_2W_3}$. Then,
	\begin{equation} \label{eq:firstTerm}
        \bra{\psihist}Z_{W_2}\otimes Z_{W_3}\ket{\psihist} = \frac{1}{L+13} \sum_{t=0}^{L+12}\trace\left((\ketbra{\phi_t}{\phi_t}_{Q,Q',W}\otimes\ketbra{t}{t}_C) Z_{W_2}\otimes Z_{W_3}\right),
	\end{equation}
since $Z_{W_2}\otimes Z_{W_3}$ acts invariantly on the clock register. Defining $\ket{v}:=\sum_{t=L+1}^{L+12}\ket{\phi_t}_{Q,Q',W}\ket{t}_C$, we have that since $W_2W_3$ is set to $\ket{00}$ for times $0\leq t\leq L$, Equation~(\ref{eq:firstTerm}) simplifies to
	\[
        \frac{1}{L+13}\left(L+1 + \bra{v}Z_{W_2}\otimes Z_{W_3}\ket{v}\right).
    \]
Thus, via similar reasoning
	\begin{eqnarray}
		f(\ket{\psihist},Z_{W_2},Z_{W_3})&=&\frac{1}{L+13}\left[(L+1) + \bra{v}Z_{W_2}\otimes Z_{W_3}\ket{v}\right] -\nonumber\\ && \frac{1}{(L+13)^2} \left[(L+1) + \bra{v}Z_{W_2}\ket{v}\right]\left[(L+1) + \bra{v}Z_{W_3}\ket{v})\right] .\label{eqn:reduced}
	\end{eqnarray}
	Suppose now that $\Pi$ is a YES instance. Then there exists a history state $\ket{\psihist}$ in the ground space of $H$ (i.e. with quantum proof $\ket{\phi}_{Q,Q'}=\ket{\hat{x}}_Q\otimes\ket{\phi'}_{Q'}$ for a correct query string $x$) for which $W_2W_3$ is set to $\ket{\phi ^+}$ in the final seven timesteps (since $U'$ is deterministic). Since $\bra{\phi^+}Z\otimes Z\ket{\phi^+}=1$ and $\bra{\phi^+}Z\otimes I\ket{\phi^+}=0$, via Equation~(\ref{eqn:reduced}) we have
	\[
        f(\ket{\psihist},Z_{W_2},Z_{W_3})\geq\frac{(L+1)-5+7}{L+13}-\frac{((L+1)+5)^2}{(L+13)^2} = \frac{1}{L+13}\left(4-\frac{49}{L+13}\right) ,
    \]
where the $\pm5$ terms correspond to timesteps $t=L+1,\ldots,L+5$ and use the fact that $\snorm{Z}=1$.
	
    Conversely, suppose $\Pi$ is a NO instance, and consider any $\ket{\psi}$ satisfying $\bra{\psi}H\ket{\psi}\leq \lambda(H)+\delta$. Then, as argued in the proof of Theorem~\ref{thm:main1}, there exists a history state $\ket{\psi''}$ on ``proof'' $\ket{\phi_1}_{Q,Q'}$ (consisting of a superposition of correct query strings) satisfying
\begin{equation}\label{eqn:13}
    \trnorm{\ketbra{\psi}{\psi}-\ketbra{\psi''}{\psi''}}
    \leq \frac{1}{2(L+13)},
\end{equation}
         by Equations~(\ref{eqn:4}), (\ref{eqn:5}) and~(\ref{eqn:6}). Since the history state $\ket{\psi''}$ has $W_2W_3$ set to $\ket{00}$ in all time steps, we know because $Z\otimes Z \ket{00}=\ket{00}$ that $f(\ket{\psi''},Z_{W_2},Z_{W_3})=0$. Thus, using Equation~(\ref{eqn:reduced}) and applying the H\"{o}lder inequality to {the second} term of $f(\ket{\psi},Z_{W_2},Z_{W_3})$ yields
    \[
        f(\ket{\psi},Z_{W_2},Z_{W_3})\leq 1-\left( 1-\frac{1}{2(L+13)}\right)^2 = \frac{1}{L+13}\left(1-\frac{1}{4(L+13)}\right) .
    \]
\end{proof}

\begin{lemma} \label{lem:2-PIn} \tpc\ is in $\PQMA$.
\end{lemma}
\begin{proof}
	The proof combines ideas from Ambainis's original proof of $\app\in\PQMA$~\cite{A14} (see Theorem 6 therein) and a trick of Chailloux and Sattath~\cite{CS11} from the study of $\class{QMA(2)}$. We give a proof sketch here. Specifically, let $\Pi=(H,A,B,k,l,a,b,\delta)$ be an instance of \tpc. Similar to~\cite{A14}, the $\PQMA$ verification procedure proceeds, at a high level, as follows:
	\begin{enumerate}
		\item Use logarithmically many QMA oracle queries to perform a binary search to obtain an estimate $\gamma\in\reals$ satisfying $\lambda(H)\in[\gamma ,\gamma +\frac{\delta}{2}]$.
		\item Use a single QMA oracle query to verify the statement: ``There exists $\ket{\psi}$ satisfying (1) $\bra{\psi}H\ket{\psi}\leq \lambda(H)+\delta$ and (2) $f(\ket{\psi},A,B)\geq a$''.
	\end{enumerate}
The first of these steps is performed identically to the proof of $\app\in\PQMA$~\cite{A14}; we do not elaborate further here. The second step, however, differs from~\cite{A14} for the following reason. Intuitively,~\cite{A14} designs a QMA protocol which takes in many copies of a proof $\ket{\psi}$, performs phase estimation on each copy, postselects to ``snap'' each copy of $\ket{\psi}$ into a low-energy state $\ket{\psi_i}$ of $H$, and subsequently uses states $\set{\ket{\psi_i}}$ to estimate the expectation against an observable $A$. If the ground space of $H$ is degenerate, the states $\set{\ket{\psi_i}}$ may not all be identical. This does not pose a problem in~\cite{A14}, as there soundness of the protocol is guaranteed since {all} low energy states have high expectation against $A$. In our setting, however, if we use this protocol to individually estimate each of the terms $\bra{\psi}A\otimes B \ket{\psi}$, $\bra{\psi}A \ket{\psi}$, and $\bra{\psi} B \ket{\psi}$, soundness \emph{can} be violated if each of these three terms are not estimated using the same state $\ket{\psi_i}$, since the promise gap of the input does not necessarily say anything about the values of each of these three terms individually.

To circumvent this, we observe that to evaluate $f(\ket{\psi},A,B)$, we do not need the ground state $\ket{\psi}$ itself, but only a classical description of its local reduced density matrices (a similar idea was used in~\cite{CS11} to verify the energy of a claimed product state proof against a local Hamiltonian in the setting of QMA(2)). Specifically, suppose $\Pi$ consists of a $k$-local Hamiltonian $H$ acting on $n$ qubits. Then, the prover sends classical descriptions of $k$-qubit density matrices $\set{\rho_S}$ for each subset $S\subseteq[n]$ of size $\abs{S}=k$, along with a QMA proof that the states $\set{\rho_S}$ are consistent with a global $n$-qubit pure state $\ket{\psi}$ (recall the problem of verifying consistency is QMA-complete~\cite{L06}). The verifier runs the QMA circuit for consistency, and assuming this check passes it uses the classical $\set{\rho_S}$ to classically verify that (1) $\bra{\psi}H\ket{\psi}\leq \lambda(H)+\delta$ and (2) $f(\ket{\psi},A,B)\geq a$ (since both of these depend only on the local states $\set{\rho_S}$).

\end{proof}

\subsection{$\PUQMA$-hardness of \spgap}\label{scn:spectralGap}

	We now restate and prove Theorem~\ref{thm:spgap}. We begin by defining $\spgap$ and \UQMA.
\begin{definition}[$\spgap(H,\alpha)$~(Ambainis~\cite{A14})]
	Given a Hamiltonian $H$ and a real number $\alpha\geq n^{-c}$ for $n$ the number of qubits $H$ acts on and $c>0$ some constant, decide:
	\begin{itemize}
		\item If $\lambda_2 - \lambda_1 \leq \alpha$, output YES.
		\item If $\lambda_2 - \lambda_1 \geq 2\alpha$, output NO.
	\end{itemize}
where $\lambda_2$ and $\lambda_1$ denote the second and first smallest eigenvalues of $H$, respectively.
\end{definition}
\noindent For clarity, if the ground space of $H$ is degenerate, then we define its spectral gap as $0$.
\begin{definition}[Unique QMA (UQMA)~(Aharonov \emph{et al.}~\cite{ABBS08})]
      We say a promise problem $A=(\ayes,\ano)$ is in Unique QMA if and only if there exist polynomials $p$, $q$ and a polynomial-time uniform family of quantum circuits $\set{Q_n}$, where $Q_n$ takes as input a string $x\in\Sigma^n$, a quantum proof $\ket{y}\in (\complex^2)^{\otimes p(n)}$, and $q(n)$ ancilla qubits in state $\ket{0}^{\otimes q(n)}$, such that:
    \begin{itemize}
    \item (Completeness) If $x\in\ayes$, then there exists a proof $\ket{y}\in (\complex^2)^{\otimes p(n)}$ such that $Q_n$ accepts $(x,\ket{y})$ with probability at least $2/3$, and for all $\ket{\hat{y}}\in (\complex^2)^{\otimes p(n)}$ orthogonal to $\ket{y}$, $Q_n$ accepts $(x,\ket{\hat{y}})$ with probability at most $1/3$.
    \item (Soundness) If $x\in\ano$, then for all proofs $\ket{y}\in (\complex^2)^{\otimes p(n)}$, $Q_n$ accepts $(x,\ket{y})$ with probability at most $1/3$.
    \end{itemize}
\end{definition}

The main theorem of this section is the following.
		
\begin{reptheorem}{thm:spgap}
	$\spgap$ is $\PUQMA$-hard for $4$-local Hamiltonians $H$ under polynomial-time Turing reductions (i.e. Cook reductions).
\end{reptheorem}

	\noindent We remark that Ambainis~\cite{A14} showed that $\spgap\in\PQMA$, and gave a claimed proof that $\spgap$ is $\PUQMA$-hard for $O(\log)$-local Hamiltonians under mapping reductions. ($\PUQMA$ is defined as $\PQMA$, except with a UQMA oracle in place of a QMA oracle.) As discussed in Section~\ref{scn:intro}, however, Ambainis' proof of the latter result does not hold if the $\PUQMA$ machine makes invalid queries (which in general is the case). Here, we build on Ambainis' approach~\cite{A14} to show $\PUQMA$-hardness of $\spgap$ under Turing reductions even when invalid queries are allowed, and we also improve the hardness to apply to $O(1)$-local Hamiltonians.

We begin by showing the following modified version of Lemma \ref{l:amb} tailored to UQMA (instead of QMA). It is crucial to observe that the mapping of Lemma~\ref{l:amb} which produces Hamiltonian $H$ is \emph{efficient}, i.e. $H$ can be computed in polynomial time. However, in Lemma~\ref{lem:spgap} below, our mapping to a Hamiltonian $H$ is \emph{not} clearly efficient, meaning the lemma only shows the \emph{existence} of $H$. Roughly, this is because the proof of Lemma~\ref{lem:spgap} proceeds by replacing invalid queries with ``dummy'' NO queries to obtain the desired spectral gap. But a polynomial-time machine alone cannot identify such invalid queries, and hence cannot simulate this mapping.

Henceforth, we assume that all calls to the \UQMA~oracle $Q$ are for an instance $(H,a,b)$ of the Unique-Local Hamiltonian Problem (U-LH) (which is complete for \UQMA~\cite{A14}): Is the ground state energy of $H$ at most $\epsilon$ with all other eigenvalues at least $3\epsilon$ (YES case), or is the ground state energy at least $3\epsilon$ (NO case), for $\epsilon\geq 1/\poly(n)$?
\begin{lemma}\label{lem:spgap}
	For any $x\in\set{0,1}^m$, let $\hat{x}$ denote its unary encoding. Then, for any $\PUQMA$ circuit $U$ acting on $n$ bits and making $m$ queries to a \UQMA~oracle, there exists a $4$-local Hamiltonian $H$ acting on space $(\complex^2)^{\otimes 2^m-1}\otimes\spa{Y}$ such that there exists a correct query string $x=x_1\cdots x_m$ such that:
	\begin{enumerate}
		\item The {unique} ground state of $H$ lies in subspace $\ketbra{\hat{x}}{\hat{x}}\otimes \spa{Y}$.
		\item The spectral gap of $H$ is at least $(\epsilon - \delta )/4^{m}$ for inverse polynomial $\epsilon,\delta$ with $\epsilon-\delta\geq 1/\poly(n)$.
		\item For all strings $x'\in\set{0,1}^m$, $H$ acts invariantly on subspace $\ketbra{\hat{x}'}{\hat{x}'}\otimes \spa{Y}$.
	\end{enumerate}
\end{lemma}

\begin{proof}
We begin by giving the construction of $H$.
\paragraph{Construction.}
	As done in~\cite{A14}, we begin with $O(\log n)$-local Hamiltonian
	\begin{equation}\label{eqn:H'}
		H' = \sum_{i=1}^m\frac{1}{4^{i-1}}\sum_{y_1,\ldots,y_{i-1}}\bigotimes_{j=1}^{i-1}\ketbra{y_j}{y_j}_{\spa{X}_{j}}\otimes G'_{y_{1}\cdots y_{i-1}},
    \end{equation}
where we define
\begin{equation}\label{eqn:HamG}
    G'_{y_{1}\cdots y_{i-1}} := \ketbra{0}{0}_{\spa{X}_i} \otimes A_{\spa{Y}_i} + \ketbra{1}{1}_{\spa{X}_i}\otimes  H_{\spa{Y}_i}^{i,y_{1}\cdots y_{i-1}}
\end{equation}
with $A$ any fixed $2$-local Hermitian operator with unique ground state of eigenvalue $2\epsilon$ and spectral gap $\epsilon$. Our approach is intuitively now as follows. We first run a \emph{query validation} phase, in which we modify $H'$ to obtain a new Hamiltonian $H''$ by {replacing} ``sufficiently invalid'' queries $H_{\spa{Y}_i}^{i,y_{1}\cdots y_{i-1}}$ with high-energy dummy queries. This creates the desired spectral gap. We then apply the technique of Lemma~\ref{l:amb} to reduce the locality of $H''$, obtaining a $4$-local Hamiltonian $H$, as desired. Note that our proof shows \emph{existence} of $H$; unlike Lemma~\ref{l:amb}, however, it is not clear how to construct $H$ in polynomial-time given $H'$, as detecting invalid UQMA queries with a P machine seems difficult.

The query validation runs as follows. Fix any $1/\poly(n)<\delta<\epsilon$ such that $\epsilon-\delta\geq 1/\poly(n)$; for example, set $\delta=\epsilon/2$. For any $G'_{y_{1}\cdots y_{i-1}}$ whose spectral gap is at most $\epsilon-\delta$, replace it with
\begin{equation}\label{eqn:newG}
    G_{y_{1}\cdots y_{i-1}} := \ketbra{0}{0}_{\spa{X}_i} \otimes A_{\spa{Y}_i} + \ketbra{1}{1}_{\spa{X}_i}\otimes  3\epsilon I
\end{equation}
in $H'$, denoting the new Hamiltonian as $H''$. Two remarks: By Case 2 of Lemma~\ref{l:hgap} below, the validation phase does not catch \emph{all} invalid queries. Second, setting $1/\poly(n)<\delta$ and $\epsilon-\delta\geq 1/\poly(n)$ (as opposed to, say, $1/\exp(n)$) is required for our proof of Theorem~\ref{thm:spgap} later.

\paragraph{Correctness.} Intuitively, the query validation phase roughly replaces each ``sufficiently invalid'' query $i$ with a valid NO query. This is formalized by Lemma~\ref{l:hgap}, shown subsequently, which should be roughly interpreted as saying a ``sufficiently invalid'' query $i$ is one for which the spectral gap of $G'_{y_{1}\cdots y_{i-1}}$ is at most $\epsilon-\delta$. For any such replaced query $i$, henceforth in this proof, a query string which answers YES (i.e. $\ket{1}$) to query $i$ is considered incorrect with respect to $H''$ (since $i$ is now a valid NO query in $H''$). Crucially, if a string $x$ is a correct query string for $H''$, then it is also a correct query string for $H'$. The converse is false; nevertheless, $H''$ has at least one correct query string (since any invalid query would have allowed both $\ket{0}$ and $\ket{1}$ as answers), which suffices for our purposes.

As in the proof of Lemma~\ref{l:amborig}, observe that $H''$ is block-diagonal with respect to register $\bigotimes_{i=1}^m\spa{X}_i$. Let $x\in\set{0,1}^m$ denote a correct query string which has minimal energy among all \emph{correct} query strings against $H''$, and for any $y\in\set{0,1}^m$, define $\lambda_y$ as the smallest eigenvalue in block $\spa{H}_y$. A similar analysis to that of Lemma~\ref{l:amborig} shows that for any incorrect query string $y$, $\lambda_y\geq \lambda_x+\epsilon/4^m$. This is because replacing the term $2\epsilon I$ in $M_{y_1\cdots y_{i-1}}$ from Lemma~\ref{l:amborig} with $A$ in $G_{y_{1}\cdots y_{i-1}}$ here preserves the property that answering NO on query $i$ yields minimum energy $2\epsilon$.
	
We now argue that $x$ is in fact \emph{unique}, and all other eigenvalues of $H''$ corresponding to correct query strings have energy at least $\lambda_x+(\epsilon - \delta )/4^{m}$. There are two cases to consider: eigenvalues arising from different query strings and eigenvalues arising from the same query string.

\paragraph*{Case 1: Eigenvalues from different query strings.} Let $y=y_1\cdots y_m\neq x$ be a correct query string for $H''$. Since both $x$ and $y$ are correct strings, there must exist an invalid query $i$ where $x_i \neq y_i$.

First consider the case where $G'_{y_{1}\cdots y_{i-1}}$ has spectral gap at most $\epsilon - \delta$. Then, after the validation phase, query $i$ is replaced with a valid NO query $G_{y_{1}\cdots y_{i-1}}$. Thus, whichever of $x$ or $y$ has a $1$ as bit $i$ is an incorrect string for $H''$, and from our previous analysis has energy at least $\lambda_x + \epsilon / 4^m$ against $H''$. (This, in particular, implies $x_i=0$ and $y_i=1$, by the minimality of $x$.)

Alternatively, suppose $G'_{y_{1}\cdots y_{i-1}} = G_{y_{1}\cdots y_{i-1}}$ has spectral gap strictly larger than $\epsilon - \delta$. By Lemma~\ref{l:hgap}, the invalid query $i$ must be a ``YES instance violating the uniqueness promise'' with corresponding query Hamiltonian $H_{\spa{Y}_i}^{i,y_{1}\cdots y_{i-1}}$ satisfying $\lambda(H_{\spa{Y}_i}^{i,y_{1}\cdots y_{i-1}})\leq \epsilon$. Thus, the $0$-block of Equation~(\ref{eqn:HamG}) has minimum eigenvalue $2\epsilon$ by construction, whereas the $1$-block of Equation~(\ref{eqn:HamG}) has minimum eigenvalue at most $\epsilon$. It follows from our previous analysis that whichever of $x$ or $y$ has a $0$ as bit $i$ has energy at least $\lambda_x + \epsilon / 4^m$ against $H''$. (This, in particular, implies $x_i=1$ and $y_i=0$, by the minimality of $x$.)

\paragraph*{Case 2: Eigenvalues from the same query string.} Let us next consider eigenvalues in the minimal block $\spa{H}_x$ with correct query string $x$. In this block, $H''$ is equivalent to operator
\[
		M:=\sum_{i=1}^m\frac{1}{4^{i-1}} B_{x_{1}\cdots x_{i-1}},
\]
where $B_{x_{1}\cdots x_{i-1}} = A$ if $x_i=0$ and $B_{x_{1}\cdots x_{i-1}}$ can equal either $H_{\spa{Y}_i}^{i,y_{1}\cdots y_{i-1}}$ or $3\epsilon I$ (depending on how the validation phase proceeded) if $x_i=1$. To show our claim, it suffices to show that $M$ has spectral gap $(\epsilon-\delta)/4^m$.

Recall now that each $B_{x_{1}\cdots x_{i-1}}$ acts non-trivially only on space $\spa{Y}_i$. Therefore, the minimum eigenvalue of $M$ is obtained by setting each space $\spa{Y}_i$ to the ground state vector $\ket{\psi_i}$ of $B_{x_{1}\cdots x_{i-1}}$. If we can now show that each $B_{x_{1}\cdots x_{i-1}}$ has a spectral gap of $\epsilon-\delta$, it will hence follow that swapping $\ket{\psi_i}$ for any excited state $\ket{\psi^\perp_i}$ orthogonal to $\ket{\psi_i}$ on $\spa{Y}_i$ yields an additive increase in energy of at least $\epsilon-\delta$ against any $B_{x_{1}\cdots x_{i-1}}$. It would then follow that the spectral gap of $M$ is at least $(\epsilon-\delta)/4^m$, as desired.

So let us argue that each $B_{x_{1}\cdots x_{i-1}}$ has spectral gap at least $\epsilon-\delta$. There are three cases to consider:
\begin{itemize}
    \item If $B_{x_{1}\cdots x_{i-1}}=A$, it has spectral gap $\epsilon$ by definition of $A$.
    \item If $B_{x_{1}\cdots x_{i-1}}=H_{\spa{Y}_i}^{i,y_{1}\cdots y_{i-1}}$, then consider first the case where query $i$ was valid. Since $B_{x_{1}\cdots x_{i-1}}=H_{\spa{Y}_i}^{i,y_{1}\cdots y_{i-1}}$ only if $x_i=1$, it follows that the correct answer to query $i$ is $1$. Thus, $H_{\spa{Y}_i}^{i,y_{1}\cdots y_{i-1}}$ is a valid YES instance of U-LH, and so has a spectral gap of at least $2\epsilon$ by definition.

        If instead query $i$ was invalid, then since $B_{x_{1}\cdots x_{i-1}}=H_{\spa{Y}_i}^{i,y_{1}\cdots y_{i-1}}$ and not $B_{x_{1}\cdots x_{i-1}}=3\epsilon I$, it must be that the query validation phase did not catch this invalid query, implying the spectral gap of $G'_{y_{1}\cdots y_{i-1}}$ was strictly larger than $\epsilon-\delta$. But then Lemma~\ref{l:hgap} implies that the  spectral gap of $H_{\spa{Y}_i}^{i,y_{1}\cdots y_{i-1}}$ must be at least $\epsilon-\delta$, as desired.
    \item Suppose $B_{x_{1}\cdots x_{i-1}}=3\epsilon I$; this happens only when the query validation phase replaced an invalid query $i$ with a valid NO query $3\epsilon I$. But since it also holds that $B_{x_{1}\cdots x_{i-1}}=3\epsilon I$ only if $x_i=1$, this implies that $x$ is not a correct query string with respect to $H''$ (since it answered $1$ on query $i$ instead of $0$), which is a contradiction.
\end{itemize}
 This completes the correctness proof.

\paragraph{Getting $H''$ down to $4$-local $H$.} Finally, the approach of Lemma~\ref{l:amb} allows us to convert $O(\log n)$-local $H''$ to $4$-local $H$.
\end{proof}

The following lemma is used in the proof of Lemma~\ref{lem:spgap} above, and assumes  the notation introduced therein.
\begin{lemma}\label{l:hgap}
    In Equation~(\ref{eqn:HamG}), suppose $H_{\spa{Y}_i}^{i,y_{1}\cdots y_{i-1}}$ is  invalid. Precisely one of the following cases holds.
    \begin{enumerate}
        \item (Neither YES or NO case) If there exists $0< \delta\leq\epsilon$ such that $\lambda(H_{\spa{Y}_i}^{i,y_{1}\cdots y_{i-1}})=\epsilon+\delta$ or $\lambda(H_{\spa{Y}_i}^{i,y_{1}\cdots y_{i-1}})=3\epsilon-\delta$, then the spectral gap of $G'_{y_{1}\cdots y_{i-1}}$ is at most $\epsilon - \delta$.
        \item (YES case violating uniqueness promise) If $\lambda(H_{\spa{Y}_i}^{i,y_{1}\cdots y_{i-1}})\leq\epsilon$ and  $\lambda_2(H_{\spa{Y}_i}^{i,y_{1}\cdots y_{i-1}})<3\epsilon$, then for any $0<\delta<\epsilon$, either the spectral gap of $G'_{y_{1}\cdots y_{i-1}}$ is at most $\epsilon - \delta$, or the spectral gap of $H_{\spa{Y}_i}^{i,y_{1}\cdots y_{i-1}}$ is least $\epsilon-\delta$.
    \end{enumerate}
\end{lemma}
\begin{proof}
    Observe that $G'_{y_{1}\cdots y_{i-1}}$ is block-diagonal with respect to register $\spa{X}_i$; we will refer to each of these blocks as the $0$- and $1$-block, respectively. For clarity, when we refer to $\lambda_2(X)$ for an operator $X$ below, if the ground space of $X$ is degenerate then we define $\lambda_2(X)=\lambda(X)$.

   \paragraph{Case 1.} Suppose first that $\lambda(H_{\spa{Y}_i}^{i,y_{1}\cdots y_{i-1}})=\epsilon+\delta$. Then, since the smallest eigenvalue in the $0$-block is $2\epsilon$, we have that $\lambda(G'_{y_{1}\cdots y_{i-1}})=\lambda(H_{\spa{Y}_i}^{i,y_{1}\cdots y_{i-1}})$, and so the spectral gap of  $G'_{y_{1}\cdots y_{i-1}}$ equals
   \[
    \min(2\epsilon,\lambda_2(H_{\spa{Y}_i}^{i,y_{1}\cdots y_{i-1}}))-(\epsilon+\delta)\leq \epsilon-\delta.
   \]
   Suppose next that $\lambda(H_{\spa{Y}_i}^{i,y_{1}\cdots y_{i-1}})=3\epsilon-\delta$. Then, $\lambda(G'_{y_{1}\cdots y_{i-1}})=2\epsilon$ by the $0$-block. Since the spectral gap of the $0$-block is $\epsilon$ by construction, we thus have that $\lambda_2(G'_{y_{1}\cdots y_{i-1}})=\lambda(H_{\spa{Y}_i}^{i,y_{1}\cdots y_{i-1}})$, and so the spectral gap of  $G'_{y_{1}\cdots y_{i-1}}$ equals $(3\epsilon-\delta)-2\epsilon=\epsilon-\delta$.

   \paragraph{Case 2.} Since the smallest eigenvalue of the $0$-block is $2\epsilon$, we have $\lambda(G'_{y_{1}\cdots y_{i-1}})=\lambda(H_{\spa{Y}_i}^{i,y_{1}\cdots y_{i-1}})$. Therefore, the spectral gap of $G'_{y_{1}\cdots y_{i-1}}$ equals
   \begin{equation}\label{eqn:14}
    \min(2\epsilon,\lambda_2(H_{\spa{Y}_i}^{i,y_{1}\cdots y_{i-1}}))-\lambda(H_{\spa{Y}_i}^{i,y_{1}\cdots y_{i-1}}).
   \end{equation}
   For any $\delta$ as defined in the claim, suppose now that the spectral gap of $G'_{y_{1}\cdots y_{i-1}}$ is at least $\epsilon-\delta$. Then if in Equation~(\ref{eqn:14}) $\min(2\epsilon,\lambda_2(H_{\spa{Y}_i}^{i,y_{1}\cdots y_{i-1}}))=2\epsilon$, since $\lambda(H_{\spa{Y}_i}^{i,y_{1}\cdots y_{i-1}})\leq \epsilon$ by assumption, we conclude the spectral gap of $H_{\spa{Y}_i}^{i,y_{1}\cdots y_{i-1}}$ is at least $\epsilon$. Otherwise, if $\min(2\epsilon,\lambda_2(H_{\spa{Y}_i}^{i,y_{1}\cdots y_{i-1}}))=\lambda_2(H_{\spa{Y}_i}^{i,y_{1}\cdots y_{i-1}})$, then the smallest two eigenvalues of $H_{\spa{Y}_i}^{i,y_{1}\cdots y_{i-1}}$ and $G'_{y_{1}\cdots y_{i-1}}$ coincide, and so the spectral gap of the former is at least $\epsilon-\delta$. Thus, we conclude that either the spectral gap of $G'_{y_{1}\cdots y_{i-1}}$ is at most $\epsilon - \delta$, or the spectral gap of $H_{\spa{Y}_i}^{i,y_{1}\cdots y_{i-1}}$ is least $\epsilon-\delta$.
\end{proof}

We now prove the main theorem of this section.
\begin{proof}[Proof of Theorem \ref{thm:spgap}]
As done in~\cite{A14}, we start with the Hamiltonian $H'$ from Equation~(\ref{eqn:H'}). In \cite{A14}, it was shown (Section A.3, Claim 2) that if all query Hamiltonians $H_{\spa{Y}_i}^{i,y_{1}\cdots y_{i-1}}$ correspond to valid UQMA queries, $H'$ has a unique ground state and spectral gap at least $\epsilon/4^m$. When invalid queries are allowed, however, the spectral gap of $H'$ can vanish, invalidating the $\PUQMA$-hardness proof of~\cite{A14}. Thus, we require a technique for identifying invalid queries and ``removing them'' from $H'$. Unfortunately, it is not clear how a P machine alone can achieve such a ``property testing'' task of checking if a query is sufficiently invalid. However, the key observation is that an oracle $Q$ for SPECTRAL GAP can help.

     We proceed as follows. Given an arbitrary $\PUQMA$ circuit $U$ acting on $n$ bits, construct $O(\log n)$-local $H'$ from Equation~(\ref{eqn:H'}). For each term $G'_{y_{1}\cdots y_{i-1}}$ appearing in $H'$, perform binary search using $O(\log n)$ queries to $Q$ to obtain an estimate $\Delta$ for the spectral gap of $G'_{y_{1}\cdots y_{i-1}}$ to within sufficiently small but fixed additive error $\delta\in 1/\poly(n)$. (A similar procedure involving a QMA oracle is used in Ambainis' proof of containment of $\app\in\PQMA$ to estimate the smallest eigenvalue of a local Hamiltonian; we hence omit further details here.) As done in the proof of Lemma~\ref{lem:spgap}, if $\Delta\leq \epsilon-\delta$, we conclude $H_{\spa{Y}_i}^{i,y_{1}\cdots y_{i-1}}$ is ``sufficiently invalid'', and replace $G'_{y_{1}\cdots y_{i-1}}$ with $G_{y_{1}\cdots y_{i-1}}$ from Equation~(\ref{eqn:newG}). Following the construction of Lemma~\ref{lem:spgap}, we hence can map $H'$ to a $4$-local Hamiltonian $H$ such that $H$ has a unique ground state and spectral gap $(\epsilon-\delta)/4^m$, and the ground state of $H$ corresponds to a correct query string for $H'$. Note that implementing the mapping from $H'$ to $H$ requires polynomially many queries to the oracle, hence yielding a polynomial-time \emph{Turing} reduction.
	
	Next, following~\cite{A14}, let $T:= \sum_{y_1 ... y_m} \bigotimes_{i=1}^m \ketbra{y_i}{y_i}\in\lin{\spa{Y}}$, where we sum over all query strings $y_1...y_m$ which cause $U$ to output $0$. Unlike~\cite{A14}, as done in Lemma~\ref{l:amb}, we apply Kitaev's unary encoding trick~\cite{KSV02} and implicitly encode the query strings in $T$ in unary. (We remark the term $\hstab$ contained in $H$ will enforce the correct unary encoding in register $\spa{X}$). Finally, introduce a single-qubit register $\spa{B}$, and define
	\[
		H_{\rm final} := I_B \otimes H_{\spa{X},\spa{Y}} + 4\epsilon \ketbra{0}{0}_B \otimes T_{\spa{X}}\otimes I_{\spa{Y}}.
	\]	
The claim now follows via an analysis similar to \cite{A14}. Let $\ket{\psi}_{\spa{X},\spa{Y}}$ denote the unique ground state of $H$, whose $\spa{X}$ register contains the (unary encoding of) a correct query string for $U$. If $U$ accepts, then $\ket{i}_{\spa{B}}\otimes\ket{\psi}_{\spa{X},\spa{Y}}$ for $i\in\set{0,1}$ are degenerate ground states of $H_{\rm final}$, implying $H_{\rm final}$ has no spectral gap. Conversely, if $U$ rejects, observe that the smallest eigenvalue of $H_{\rm final}$ lies in the $\ket{1}_{\spa{B}}$ block of $H_{\rm final}$. This is because $H_{\rm final}$ is block-diagonal with respect to register $\spa{X}$, and we have from the proof of Lemma~\ref{l:amb} that $\lambda(H)< 3\epsilon$. Restricted to this $\ket{1}_{\spa{B}}$ block, the spectral gap of $H_{\rm final}$ is at least $(\epsilon-\delta)/4^m$ by Lemma~\ref{lem:spgap}. Alternatively, restricted to the $\ket{0}_\spa{B}$ block, any correct query string in $\spa{X}$ leads to spectral gap at least $4\epsilon$ (by construction of $T$, since $U$ outputs $0$ in this case), and any incorrect query string in $\spa{X}$ leads to spectral gap at least $(\epsilon-\delta)/4^m$ by Lemma~\ref{lem:spgap}. Hence, $H_{\rm final}$ has an inverse polynomial spectral gap, as desired.
\end{proof}

\section{Upper bounds (containment in PP)}\label{scn:PP}

We now restate and prove Theorem~\ref{thm:inPP}.  Our approach is to use the strong error reduction technique of Marriott and Watrous~\cite{MW05} to develop a variant of the hierarchical voting scheme used in the proof of $\PNP\subseteq \class{PP}$~\cite{BHW89}. We also require a more involved analysis than present in~\cite{BHW89}, since QMA is a class of promise problems, not decision problems.

\begin{reptheorem}{thm:inPP}
    $\PQMA\subseteq\class{PP}$.
\end{reptheorem}

\begin{proof}
    Let $\Pi$ be a P machine which makes $m=c\log n$ queries to an oracle for $2$-LH, for $c\in O(1)$ and $n$ the input size.
    Without loss of generality, we assume all queries involve Hamiltonians on $M$ qubits ($M$ some fixed polynomial in $n$). Define $q:= (M+2)m$. We give a PQP computation simulating $\Pi$ (i.e. the PQP computation, which does not have access to a $2$-LH oracle, will accept if and only if $\Pi$ is a YES instance). Since $\class{PQP}=\class{PP}$~\cite{W13}, this yields the claim. Let $\abs{y}$ denote the non-negative integer with binary encoding $y$, and let $V$ denote the verification circuit for $2$-LH. The PQP computation is (intuition to follow):
    \begin{enumerate}
        \item For $i$ from $1$ to $m$:
        \begin{enumerate}
            \item Prepare $\rho=I/2^M\in\density{(\complex^2)^{\otimes M}}$.
            \item Run $V$ on the $i$-th query Hamiltonian $H_{\spa{Y}_i}^{i,y_{1}\cdots y_{i-1}}$ (see Equation~(\ref{eqn:amb1})) and proof $\rho$, and measure the output qubit in the standard basis. Set bit $y_i$ to the result.
        \end{enumerate}
        \item Let $y=y_1\cdots y_m$ be the concatenation of bits set in Step 1(b).
        \item For $i$ from $1$ to $n^c-1$:
        \begin{enumerate}
            \item If $\abs{y}< i$, then with probability $1-2^{-q}$, set $y=\#$, and with probability $2^{-q}$, leave $y$ unchanged.
        \end{enumerate}
        \item If $y=\#$, output a bit uniformly at random. Else, run $\Pi$ on query string $y$ and output $\Pi$'s answer.
    \end{enumerate}
    \textbf{Intuition.} In Step 1, one tries to determine the correct answer to query $i$ by guessing a satisfying quantum proof for verifier $V$. Suppose for the moment that $V$ has zero error, i.e. has completeness $1$ and soundness $0$, and that $\Pi$ only makes valid queries. Then, if Step 1(b) returns $y_i=1$, one knows with certainty that the query answer should be $1$. And, if the correct answer to query $i$ is $0$, then Step 1(b) returns $y_i=0$ with certainty. Thus, analogous to the classical case of an NP oracle (as done in~\cite{BHW89}), it follows that the lexicographically \emph{largest} query string $y^*$ obtainable by this procedure must be the (unique) correct query string (note that $y^*\neq 1^m$ necessarily\footnote{Under the assumptions that $V$ has zero error and $\Pi$ makes only valid queries, $y^*=1^m$ can only be obtained by this procedure if all queries are for YES instances of $2$-LH. If, on the other hand, query $i$ is a NO query, then a correct proof cannot be guessed (since it does not exist), and so $y^*_i=0$ necessarily.}). Thus, ideally one wishes to obtain $y^*$, simulate $\Pi$ on $y^*$, and output the result. To this end, Step 3 ensures that among all values of $y\neq \#$, $y^*$ is more likely to occur than all other $y\neq y^*$ combined. We now make this intuition rigorous (including in particular the general case where $V$ is not zero-error and $\Pi$ makes invalid queries).\\

    \noindent \textbf{Correctness.} To analyze correctness of our PQP computation, it will be helpful to refine our partition of the set of query strings $\set{0,1}^m$ into three sets:
    \begin{itemize}
        \item \textbf{(Correct query strings)} Let $A\subseteq\set{0,1}^m$ denote the set of query strings which correspond to correctly answering each of the $m$ queries. Note we may have $\abs{A}>1$ if invalid queries are made.
        \item \textbf{(Incorrect query strings)} Let $B\subseteq\set{0,1}^m\setminus A$ denote the set of query strings such that for any $y\in B$, all bits of $y$ which encode an incorrect query answer are set to $0$ (whereas the correct query answer would have been $1$, i.e. we failed to ``guess'' a good proof for this query in Step 1).
        \item \textbf{(Strongly incorrect query strings)} Let $C=\set{0,1}^m\setminus (A\cup B)$ denote the set of query strings such that for any $y\in C$, at least one bit corresponding to an incorrect query answer is set to $1$ (whereas the correct query answer would have been $0$). Such an error can only arise due to the bounded-error of our QMA verifier in Step 1(b).
    \end{itemize}

    Let $Y$ be a random variable corresponding to the query string $y$ obtained at the end of Step 3. To show correctness, we claim that it suffices to show that
    $
        \Delta:=\pr[Y\in A]-\pr[Y\in B\cup C]>0.
    $
    To see this, let $p_1$, $p_2$, and $p_3$ denote the probability that after Step 3, $y=\#$, $y\in A$, and $y\in B\cup C$, respectively. Then, $p_1+p_2+p_3=1$, and let $p_2-p_3=\Delta>0$. Suppose now that the input to $\Pi$ is a YES instance. Then, our protocol outputs $1$ with probability at least
    $
        \frac{p_1}{2}+p_2=\frac{1-p_2-p_3}{2}+p_2=\frac{1+\Delta}{2}>\frac{1}{2}.
    $
    If the input is a NO instance, the protocol outputs $1$ with probability
    $
        \leq \frac{p_1}{2}+p_3=\frac{1-\Delta}{2}<\frac{1}{2}.
    $
    We hence have a PQP computation, as desired. We thus now show that $\Delta>0$.

    To ease the presentation, we begin by making two assumptions (to be removed later): (i) $V$ is zero-error and (ii) $\Pi$ makes only valid queries. In this case, assumption (i) implies $C=\emptyset$ (i.e. all incorrect query strings belong to $B$), and (ii) implies $A$ is a singleton (i.e. there is a unique correct query string $y^*$). Thus, here $\Delta=\pr[Y\in A]-\pr[Y\in B]$.

    To begin, note that for any $y\in\set{0,1}^m$, we have\vspace{-1mm}
    \begin{equation}\label{eqn:prob}
        \pr[Y=y]=\pr[y \text{ chosen in Step 2 }] \cdot \left(\frac{1}{2^q}\right)^{(n^c-1)-\abs{y}} .
    \end{equation}
    \noindent Let $\hw(x)$ denote the Hamming weight of $x\in\set{0,1}^m$. Since each query corresponds to a verifier on $M$ proof qubits, we have for (the unique) $y^*\in A$ that\vspace{-1mm}
    \begin{equation}\label{eqn:LB}
        \pr[y^* \text{ chosen in Step 2 }]\geq 2^{-M\cdot\hw(y^*)}\geq 2^{-Mm}
    \end{equation}
    (recall from Section~\ref{scn:preliminaries} that setting $\rho=I/2^M$ simulates ``guessing'' a correct proof with probability at least $1/2^M$). It follows by Equations~(\ref{eqn:prob}) and~(\ref{eqn:LB}) that
    \begin{eqnarray}
          \Delta &\geq&\left(\frac{1}{2^q}\right)^{(n^c-1)-\abs{y^*}}\left[\frac{1}{2^{Mm}}-\sum_{y\in B}\pr[y \text{ chosen in Step 2 }]\cdot \left(\frac{1}{2^q}\right)^{\abs{y^*}-\abs{y}}\right]\nonumber\\
          &\geq&\left(\frac{1}{2^q}\right)^{(n^c-1)-\abs{y^*}}\left[\frac{1}{2^{Mm}}- (2^m)\left(\frac{1}{2^q}\right)\right]
          \geq\left(\frac{1}{2^q}\right)^{(n^c-1)}\frac{1}{2^{Mm}}\left[1- \frac{1}{2^{m}}\right],\label{eqn:zeroerror}
    \end{eqnarray}
    where the second inequality follows from the trivial bound $\pr[y \text{ chosen in Step 2 }]\leq 1$ and since $y\in B$ if and only if $\abs{y}<\abs{y^*}$, and the third since $q=(M+2)m$. Thus, $\Delta>0$ as desired.\\

    \noindent\textbf{Removing assumption (i).} We now remove the assumption that $V$ is zero error. In this case, $A$ is still a singleton; let $y^*\in A$. We can now also have strongly incorrect query strings, i.e. $C\neq \emptyset$ necessarily. Assume without loss of generality that $V$ acts on $M$ proof qubits, and by strong error reduction~\cite{MW05} has completeness $c:=1-2^{-p(n)}$ and soundness $s:=2^{-p(n)}$, for $p$ a polynomial to be chosen as needed. Then, since $V$ can err, Equation~(\ref{eqn:LB}) becomes
    \begin{eqnarray}
        \pr[y^*\text{ chosen in Step 2 }]&\geq&\left(\frac{c}{2^M}\right)^{\hw(y^*)}\left(1-s\right)^{m-\hw(y^*)}=\frac{1}{2^M}^{\hw(y^*)}e^{m\ln(1-\frac{1}{2^p})}\nonumber\\
        &\geq& \frac{1}{2^{Mm}}\left(1-\frac{m}{2^p-1}\right),\label{eqn:8}
    \end{eqnarray}
    where the equality follows by the definitions of $c$ and $s$, and the second inequality by applying the Maclaurin series expansion $\ln(1+x)=\sum_{n=1}^{\infty}(-1)^{n+1}\frac{x^n}{n}$ for $\abs{x}<1$ and the fact that $e^t\geq 1+t$ for all $t\in \reals$. Thus, the analysis of Equation~(\ref{eqn:zeroerror}) yields that
    \begin{equation}\label{eqn:7}
        \pr[Y\in A]-\pr[Y\in B]\geq \left(\frac{1}{2^q}\right)^{(n^c-1)}\frac{1}{2^{Mm}}\left[1- \frac{1}{2^{m}}-\frac{m}{2^p-1}\right],
    \end{equation}
    i.e. the additive error introduced when assumption $(i)$ is dropped scales as $\approx2^{-p}$. Crucially, Equation~(\ref{eqn:7}) holds for all $y\in B$ even with assumption (i) dropped since the analysis of Equation~(\ref{eqn:zeroerror}) used only the trivial bound $\pr[y \text{ chosen in Step 2 }]\leq 1$ for any $y\in B$.

     Next, we upper bound the probability of obtaining $y\in C$ in Step 2. Next, we upper bound the probability of obtaining $y\in C$ in Step 2.
     \begin{lemma}\label{l:last}
        $\pr(Y\in C)\leq \frac{2^m}{2^p}$.
     \end{lemma}
\begin{proof}
    Any $y\in C$ must have a bit $j$ incorrectly set to $1$, whereas the correct query answer (given bits $1$ through $j-1$ of $y$) should have been $0$. The probability of this occurring for bit $j$ in Step 1(b) is at most $2^{-p}$, by the soundness property of $V$. Since $\abs{C}\leq 2^m$, the claim follows.
\end{proof}
     Lemma~\ref{l:last} now implies
    \begin{equation}\label{eqn:Delta2}
        \Delta\geq\left(\frac{1}{2^q}\right)^{(n^c-1)}\frac{1}{2^{Mm}}\left[1- \frac{1}{2^{m}}-\frac{m}{2^p-1}\right] - \frac{2^m}{2^p}.
    \end{equation}
    We conclude that setting $p$ to a sufficiently large fixed polynomial ensures $\Delta>0$, as desired.\\

    \noindent\textbf{Removing assumption (ii).}  We now remove the assumption that $\Pi$ only makes valid queries, which is the most involved step. Here, $A$ is no longer necessarily a singleton. The naive approach would be to let $y^*$ denote the \emph{lexicographically largest} string in $A$, and attempt to run a similar analysis as before. Unfortunately, this no longer necessarily works for the following reason. For any invalid query $i$, we do not have strong bounds on the probability that $V$ accepts in Step 1(b); in principle, this value can lie in the range $(2^{-p},1-2^{-p})$. Thus, running the previous analysis with the lexicographically largest $y^*\in A$ may cause Equation~(\ref{eqn:Delta2}) to yield a negative quantity.
    We hence require a more delicate analysis.

    We begin by showing the following lower bound.
    \begin{lemma}\label{l:LB}
    Define $\Delta':=\pr[Y\in A]-\pr[Y \in B]$. Then,
    \[
        \Delta'\geq \left(\frac{1}{2^q}\right)^{(n^c-1)}\frac{1}{2^{Mm}}\left[1- \frac{1}{2^{m}}-\frac{m}{2^p-1}\right].
    \]
    \end{lemma}
    \begin{proof}
            For any string $y\in\set{0,1}^m$, let $I_y\subseteq\set{1,\ldots, m}$ denote the indices of all bits of $y$ set by invalid queries. We call each such $i\in I_y$ a \emph{divergence point}. (This name is chosen since on any invalid query, the computation can diverge into either of two correct computation paths.) Let $p_{y,i}$ denote the probability that (invalid) query $i$ (defined given answers to queries $1$ through $i-1$) outputs bit $y_i$, i.e. $p_{y,i}$ denotes the probability that at divergence point $i$, we go in the direction of bit $y_i$. We define the \emph{divergence probability} of $y\in\set{0,1}^m$ as $p_y = \Pi_{i\in I_{y}}p_{y,i}$, i.e. $p_y$ is the probability of answering all invalid queries as $y$ did.

The proof now proceeds by giving for each $1\leq i\leq \abs{A}$, which denotes the \emph{iteration number}, a $3$-tuple $(y_{i-1}^*, y_i^*, B_{y_i^*})\in \set{0,1}^m\times\set{0,1}^m\times \mathcal{P}(B)$, where $\mathcal{P}(X)$ denotes the power set of set $X$. Set
            \[
                \Delta'_i:=\pr[Y\in\set{y_1^*,\ldots,y_i^*}]-\pr[Y\in B_{y_1^*}\cup\cdots\cup B_{y_i^*}],
            \]
            where it will be the case that $\set{B_{y_i^*}}_{i=1}^{\abs{A}}$ is a partition of $B$. Thus, we have $\Delta'\geq \Delta'_{\abs{A}}$, implying that a lower bound on $\Delta'_{\abs{A}}$ suffices to prove our claim. We hence prove via induction on the iteration number $i$ that for all $1\leq i\leq \abs{A}$,
            \[
                \Delta'_i\geq \left(\frac{1}{2^q}\right)^{(n^c-1)}\frac{1}{2^{Mm}}\left[1- \frac{1}{2^{m}}-\frac{m}{2^p-1}\right].
            \]
            \\\vspace{-2mm}

            \noindent\emph{Base case (i=1).} In this case $y_0^*$ is undefined. Set $y_1^*$ to any string in $A$ with divergence probability  at least
    \begin{equation}\label{eqn:9}
        p_1^* = \prod_{i\in I_{y_1^*}}p_{y_1^*,i}\geq {2^{-\abs{I_{y_1^*}}}}.
    \end{equation}
    Such a string must exist, since at each divergence point $i$, at least one of the outcomes in $\set{0,1}$ occurs with probability at least $1/2$. (Note that queries are not being made to a QMA oracle in hierarchical voting (since our PQP machine does not have access to a QMA oracle), but to a QMA verifier $V$ with a maximally mixed proof as in Step 1(a). Whereas in the former case the output of the oracle on an {invalid} query does not have to consistently output a value with any particular probability, in the latter case, there is some fixed probability $p$ with which $V$ outputs $1$ each time it is run on a fixed proof.) Finally, define $B_{y_1^*}:=\set{y\in B\mid \abs{y}<\abs{y_1^*}}$ (recall $\abs{y}$ is the non-negative integer with binary encoding $y$).

    Let $k_*$ denote the number of divergence points of $y_1^*$ (i.e. $k_*=\abs{I_{y_1^*}}$), and $k_0$ ($k_1$) the number of zeros (ones) of $y_1^*$ arising from valid queries. Thus, $k_*+k_0+k_1=m$. Then, Equation~(\ref{eqn:8}) becomes
    \begin{align}
        \pr[y_1^*\text{ in Step 2 }] &\geq \left(\frac{c}{2^M}\right)^{k_1}\left(1-s\right)^{k_0}p_1^*
\nonumber
        \geq\left(\frac{1}{2^M}\right)^{k_1}\left(\frac{1}{2}\right)^{k_*}\left(1-\frac{m-k_*}{2^p-1}\right)\\ &\geq \frac{1}{2^{Mm}}\left(1-\frac{m}{2^p-1}\right),\label{eqn:10}
    \end{align}
    where the second inequality follows from Equation~(\ref{eqn:9}), and the third since $k_*\geq 0$ and $k_1+k_*\leq m$. Thus, $\Delta'_1$ is lower bounded by the expression in Equation~(\ref{eqn:7}) via an analogous analysis for $y_1^*$ and $B_{y_1^*}$.\\

\noindent\emph{Inductive step.} Assume the claim holds for $1\leq i-1<\abs{A}$. We show it holds for $i$. Let $y_{i-1}^*$ be the choice of $y^*$ in the previous iteration $i-1$. Define $A_{y_i^*}:=\set{y\in A\mid \abs{y}>\abs{y_{i-1}^*}}$. Partition $A_{y_i^*}$ into sets $S_{k}$ for $k\in[m]$, such that $S_k$ is the subset of strings in $A_{y_i^*}$ which agrees with $y_{i-1}^*$ on the first $k-1$ bits, but disagrees on bit $k$. Note that if $S_k\neq\emptyset$, then bit $k$ of $y_{i-1}^*$ is $0$ and bit $k$ of any string in $S_k$ is $1$. For each $S_k\neq\emptyset$, choose  arbitrary representative $z_k\in S_k$, and define \emph{bounded} divergence probability
\[
    q_{i}(k):=\prod_{t\in I^{\leq k}_{z_k}}p_{z_k,t}
\]
where $I^{\leq k}_{z_k}:= \set{t\in I_{z_k}\mid t \leq k}$. Note that $q_i(k)>0$ (since $S_k\neq\emptyset$). Else if $S_k = \emptyset$, set $q_i(k)=0$. Let $q_i^*$ be the max such bounded divergence probability:
\begin{equation}\label{eqn:q}
    q_i^*=\max_{k\in[m]} q_i(k) \qquad\qquad\text{and}\qquad\qquad k_i^*=\argmax_{k\in[m]} q_i(k).
\end{equation}
Let $y_i^*$ be the lexicographically largest query string in $S_{k_i^*}$ with divergence probability $p_i^*$ s.t.:
\begin{equation}\label{eqn:99}
        p_i^* \geq q_i^*\cdot 2^{-\abs{I_{y_i^*}}+\abs{I^{\leq {k_i^*}}_{y_{i}^*}}}.
\end{equation}
 That such a $y_i^*\in S_{k_i^*}$ exists follows from an argument similar to Equation~(\ref{eqn:9}): By definition, $q_i^*$ denotes the bounded divergence probability for all invalid queries up to and including query $k_i^*$, and the term exponential in $\left(-\abs{I_{y_i^*}}+\abs{I^{\leq {k_i^*}}_{y_{i}^*}}\right)$ is obtained by greedily choosing, for all invalid queries of $y_i^*$ \emph{after} query $k_i^*$, the outcome which occurs with probability at least $1/2$. Set $B_{y_i^*}:=\set{y\in B\mid \abs{y^*_{i-1}}<\abs{y}<\abs{y_i^*}}$.  The following lemma will be useful.
\begin{lemma}\label{l:LB2}
    For any $y\in B_{y_i^*}$, $\pr[y\text{ chosen in Step 2}]\leq q_i^*$, where recall $q_i^*$ is the probability from Equation~(\ref{eqn:q}).
\end{lemma}
\begin{proof}
    Fix any $y\in B_{y_i^*}$. Since $\abs{y}>\abs{y_{i-1}^*}$, there must be an index $k$ such that the $k$-th bit of $y$ is $1$ and that of $y_{i-1}^*$ is $0$. Let $k$ denote the first such index. Since $y\not\in C$ (because $B_{y_i^*}\cap C=\emptyset$), it must be that query $k$ (defined given bits $y_1\cdots y_{k-1}$) is invalid. Thus, bit $k$ is a divergence point of $y_{i-1}^*$,  and there exists a correct query string $y'\in S_k$. By Equation~(\ref{eqn:q}), $q_i^*$ was chosen as the maximum over all bounded diverge probabilities. Thus, $q_i^*\geq q_i(k)$, where recall $q_i(k)$ is the bounded divergence probability for $S_k$, where $y'\in S_k$. But since $y$ and $y'$ agree on bits $1$ through $k$ inclusive, we have $\pr[y\text{ chosen in Step 2}]\leq\prod_{t\in I^{\leq k}_{y}}p_{y,t}=q_i(k)$, from which the claim follows.
\end{proof}

    To continue with the inductive step, again consider $k_*$, $k_0$, and $k_1$, now corresponding to $y_i^*$. Then, an argument similar to Equation~(\ref{eqn:10}) says $\pr[y_i^*\text{ chosen in Step 2 }]$ is at least
        \begin{align}
        \left(\frac{c}{2^M}\right)^{k_1}\left(1-s\right)^{k_0}p_i^* &\geq \left(\frac{1}{2^M}\right)^{k_1}\left(1-\frac{m-k_*}{2^p-1}\right)q_i^*\left(\frac{1}{2}\right)^{{\abs{I_{y_i^*}}-\abs{I^{\leq {k_i^*}}_{y_{i}^*}}}} \nonumber\\
        &\geq\frac{q_i^*}{2^{Mm}}\left(1-\frac{m}{2^p-1}\right),\label{eqn:12}
    \end{align}
        where the first inequality follows from Equation~(\ref{eqn:99}), and the second since ${\abs{I_{y_i^*}}-\abs{I^{\leq {k_i^*}}_{y_{i}^*}}}\leq k_*$. Now, define $\zeta_i:= \pr[Y=y_i^*]-\pr[Y\in B_{y_i^*}]$. Applying the argument of Equation~(\ref{eqn:zeroerror}) yields
        \[
            \zeta_i\geq \left(\frac{1}{2^q}\right)^{(n^c-1)-\abs{y_i^*}}\left[\frac{q^*_i}{2^{Mm}}\left(1-\frac{m}{2^p-1}\right)-q^*_i\sum_{y\in B_{y^*_i}}\left(\frac{1}{2^q}\right)^{\abs{y_i^*}-\abs{y}}\right],
        \]
                where the first $q_i^*$ is due to Equation~(\ref{eqn:12}), and the second $q_i^*$ to Lemma~\ref{l:LB2}. Thus, similar to Equation~(\ref{eqn:7}),
        \[
            \zeta_i \geq \left(\frac{1}{2^q}\right)^{(n^c-1)}\frac{q_i^*}{2^{Mm}}\left[1- \frac{1}{2^{m}}-\frac{m}{2^p-1}\right]>0.
        \]
        Observing the recurrence that for all $i$, $\Delta'_{i}\geq\Delta'_{i-1}+\zeta_i$, unrolling this recurrence yields $\Delta'_i\geq \Delta_1$, which by the base case yields the claim.
    \end{proof}
    Finally, combining Lemmas~\ref{l:last} and~\ref{l:LB} (recall Lemma~\ref{l:last} shows $\pr(Y\in C)\leq \frac{2^m}{2^p}$; note the proof of Lemma~\ref{l:last} still applies after assumption (ii) is dropped) yields that $\pr[Y\in A]-\pr[Y \in B\cup C]$ is lower bounded by
    \[
       \pr[Y\in A]-\pr[Y \in B]-\pr[Y\in C] \geq \left(\frac{1}{2^q}\right)^{(n^c-1)}\frac{1}{2^{Mm}}\left[1- \frac{1}{2^{m}}-\frac{m}{2^p}\right] - \frac{2^m}{2^p}.
    \]
	For sufficiently large fixed $p$, this quantity is strictly positive, yielding Theorem~\ref{thm:inPP}.
\end{proof}

\section*{Acknowledgments}
We thank Xiaodi Wu for stimulating discussions which helped motivate this project, including suggesting to think about two-point correlation functions (which arose via discussions with Aram Harrow, whom we also thank). We also thank Andris Ambainis and Norbert Schuch for helpful discussions, and remark they independently conceived of some of the ideas behind Lemma~\ref{l:amb} and Theorem~\ref{thm:main1}, respectively (private communication). We thank an anonymous referee for pointing out a minor error in the proof of Lemma~\ref{l:LB} in a previous version of this draft. Part of this work was completed while SG was supported by a Government of Canada NSERC Banting Postdoctoral Fellowship and the Simons Institute for the Theory of Computing at UC Berkeley. SG acknowledges support from NSF grants CCF-1526189 and CCF-1617710. JY was supported by a Virginia Commonwealth University Presidential Scholarship.

\bibliographystyle{plainnat}
\bibliography{GY16_arxiv_v3_bibliography}

\begin{thebibliography}{63}
\providecommand{\natexlab}[1]{#1}
\providecommand{\url}[1]{\texttt{#1}}
\expandafter\ifx\csname urlstyle\endcsname\relax
  \providecommand{\doi}[1]{doi: #1}\else
  \providecommand{\doi}{doi: \begingroup \urlstyle{rm}\Url}\fi

\bibitem[Aaronson(2010)]{Aa10}
S.~Aaronson.
\newblock {BQP} and the polynomial hierarchy.
\newblock In \emph{Proceedings of the 42nd {ACM} {S}ymposium on the {T}heory of
  {C}omputing (STOC 2010)}, pages 141--150, 2010.
\newblock \doi{10.1145/1806689.1806711}.

\bibitem[Aharonov and Naveh(2002)]{AN02}
D.~Aharonov and T.~Naveh.
\newblock Quantum {NP} - {A} survey.
\newblock Available at arXiv:quant-ph/0210077v1, 2002.

\bibitem[Aharonov and Zhou(2018)]{AZ18}
D.~Aharonov and L.~Zhou.
\newblock Hamiltonian sparsification and gap-simulations.
\newblock In Avrim Blum, editor, \emph{10th Innovations in Theoretical Computer
  Science Conference (ITCS)}, volume 124 of \emph{Leibniz International
  Proceedings in Informatics (LIPIcs)}, pages 2:1--2:21. Dagstuhl Publishing,
  2018.
\newblock \doi{10.4230/LIPIcs.ITCS.2019.2}.

\bibitem[Aharonov et~al.(2008)Aharonov, Ben-Or, Brand{\~{a}}o, and
  Sattath]{ABBS08}
D.~Aharonov, M.~Ben-Or, F.~Brand{\~{a}}o, and O.~Sattath.
\newblock The pursuit for uniqueness: {E}xtending {V}aliant-{V}azirani theorem
  to the probabilistic and quantum settings.
\newblock Available at arXiv:0810.4840v1 [quant-ph], 2008.

\bibitem[Ambainis(2014)]{A14}
A.~Ambainis.
\newblock On physical problems that are slightly more difficult than {QMA}.
\newblock In \emph{Proceedings of 29th IEEE Conference on Computational
  Complexity (CCC)}, pages 32--43, 2014.
\newblock \doi{10.1109/ccc.2014.12}.

\bibitem[Bausch and Crosson(2018)]{BC18}
J.~Bausch and E.~Crosson.
\newblock Analysis and limitations of modified circuit-to-{H}amiltonian
  constructions.
\newblock \emph{Quantum}, 2:\penalty0 94, 2018.
\newblock \doi{10.22331/q-2018-09-19-94}.

\bibitem[Bausch et~al.(2017)Bausch, Cubitt, and Ozols]{BCO17}
J.~Bausch, T.~Cubitt, and M.~Ozols.
\newblock The complexity of translationally invariant spin chains with low
  local dimension.
\newblock \emph{Annales Henri Poincar\'{e}}, 18\penalty0 (11):\penalty0
  3449--3513, 2017.
\newblock \doi{10.1007/s00023-017-0609-7}.

\bibitem[Bausch et~al.(2018)Bausch, Cubitt, Lucia, and Perez-Garcia]{BCLP18}
J.~Bausch, T.~Cubitt, A.~Lucia, and D.~Perez-Garcia.
\newblock Undecidability of the spectral gap in one dimension.
\newblock Available at arXiv:1810.01858v1 [quant-ph], 2018.

\bibitem[Beigel et~al.(1989)Beigel, Hemachandra, and Wechsung]{BHW89}
R.~Beigel, L.~A. Hemachandra, and G.~Wechsung.
\newblock On the power of probabilistic polynomial time: { ${\rm
  P^{NP[log]}\subseteq PP}$}.
\newblock In \emph{Proceedings of the 4th IEEE Conference on Structure in
  Complexity Theory}, pages 225--227, 1989.
\newblock \doi{10.1109/sct.1989.41828}.

\bibitem[Bookatz(2014)]{Boo14}
A.~D. Bookatz.
\newblock {QMA}-complete problems.
\newblock \emph{Quantum Information \& Computation}, 14\penalty0
  (5--6):\penalty0 361--383, 2014.
\newblock \doi{10.26421/QIC14.5-6}.

\bibitem[Bravyi and Gosset(2015)]{BG15}
S.~Bravyi and D.~Gosset.
\newblock Gapped and gapless phases of frustration-free spin-1/2 chains.
\newblock \emph{Journal of Mathematical Physics}, 56:\penalty0 061902, 2015.
\newblock \doi{10.1063/1.4922508}.

\bibitem[Bravyi and Gosset(2017)]{BG17}
S.~Bravyi and D.~Gosset.
\newblock Complexity of quantum impurity problems.
\newblock \emph{Communications in Mathematical Physics}, 356\penalty0
  (2):\penalty0 451--500, 2017.
\newblock \doi{10.1007/s00220-017-2976-9}.

\bibitem[Bravyi and Hastings(2017)]{BH17}
S.~Bravyi and M.~Hastings.
\newblock On complexity of the quantum {I}sing model.
\newblock \emph{Communications in Mathematical Physics}, 349\penalty0
  (1):\penalty0 1--45, 2017.
\newblock \doi{10.1007/s00220-016-2787-4}.

\bibitem[Breuckmann and Terhal(2014)]{BT14}
N.~Breuckmann and B.~Terhal.
\newblock Space-time circuit-to-{H}amiltonian construction and its
  applications.
\newblock \emph{Journal of Physics A: Mathematical and Theoretical},
  47\penalty0 (19):\penalty0 195304, 2014.
\newblock \doi{10.1088/1751-8113/47/19/195304}.

\bibitem[Brown et~al.(2011)Brown, Flammia, and Schuch]{BFS11}
B.~Brown, S.~Flammia, and N.~Schuch.
\newblock Computational difficulty of computing the density of states.
\newblock \emph{Physical Review Letters}, 107\penalty0 (4):\penalty0 040501,
  2011.
\newblock \doi{10.1103/physrevlett.107.040501}.

\bibitem[Buss and Hay(1991)]{BH91}
S.~Buss and L.~Hay.
\newblock On truth-table reducibility to {SAT}.
\newblock \emph{Information and Computation}, 91\penalty0 (1):\penalty0 86 --
  102, 1991.
\newblock \doi{10.1016/0890-5401(91)90075-D}.

\bibitem[Caha et~al.(2018)Caha, Landau, and Nagaj]{CLN18}
L.~Caha, Z.~Landau, and D.~Nagaj.
\newblock Clocks in {F}eynman's computer and {K}itaev's local {H}amiltonian:
  Bias, gaps, idling, and pulse tuning.
\newblock \emph{Phys. Rev. A}, 97\penalty0 (6):\penalty0 062306, 2018.
\newblock \doi{10.1103/PhysRevA.97.062306}.

\bibitem[Castro and Seara(1992)]{CS05}
J.~Castro and C.~Seara.
\newblock Characterizations of some complexity classes between $\theta_2^p$ and
  ${\Delta}_2^p$.
\newblock In \emph{Proceedings of the 9th Annual Symposium on Theoretical
  Aspects of Computer Science (STACS)}, Lecture Notes in Computer Science,
  pages 303--317, 1992.

\bibitem[Chailloux and Sattath(2012)]{CS11}
A.~Chailloux and O.~Sattath.
\newblock The complexity of the separable {H}amiltonian problem.
\newblock In \emph{Proceedings of 27th IEEE Conference on Computational
  Complexity (CCC)}, pages 32--41, 2012.
\newblock \doi{10.1109/ccc.2012.42}.

\bibitem[Chen(2016)]{C16}
L.~Chen.
\newblock A note on oracle separations for {BQP}.
\newblock Available at arXiv:1605.00619v1 [quant-ph], 2016.

\bibitem[Cook(1972)]{C72}
S.~Cook.
\newblock The complexity of theorem proving procedures.
\newblock In \emph{Proceedings of the 3rd ACM Symposium on Theory of Computing
  (STOC)}, pages 151--158, 1972.
\newblock \doi{10.1145/800157.805047}.

\bibitem[Cubitt and Montanaro(2016)]{CM16}
T.~Cubitt and A.~Montanaro.
\newblock Complexity classification of local {H}amiltonian problems.
\newblock \emph{SIAM J. Comput.}, 45\penalty0 (2):\penalty0 268--316, 2016.
\newblock \doi{10.1137/140998287}.

\bibitem[Cubitt et~al.(2015)Cubitt, Perez-Garcia, and Wolf]{CPW15}
T.~Cubitt, D.~Perez-Garcia, and M.~M. Wolf.
\newblock Undecidability of the spectral gap.
\newblock \emph{Nature}, 528:\penalty0 207--211, 2015.
\newblock \doi{10.1038/nature16059}.

\bibitem[Cubitt et~al.(2018)Cubitt, Montanaro, and Piddock]{CMP18}
T.~Cubitt, A.~Montanaro, and S.~Piddock.
\newblock Universal quantum {H}amiltonians.
\newblock \emph{Proceedings of the National Academy of Sciences}, 115\penalty0
  (38):\penalty0 9497--9502, 2018.
\newblock \doi{10.1073/pnas.1804949115}.

\bibitem[Fefferman and Lin(2018)]{FL18}
B.~Fefferman and C.~Lin.
\newblock A complete characterization of unitary quantum space.
\newblock In A.~R. Karlin, editor, \emph{9th Innovations in Theoretical
  Computer Science Conference (ITCS)}, volume~94 of \emph{Leibniz International
  Proceedings in Informatics (LIPIcs)}, pages 4:1--4:21. Dagstuhl Publishing,
  2018.
\newblock \doi{10.4230/LIPIcs.ITCS.2018.4}.

\bibitem[Fefferman et~al.(2012)Fefferman, Shaltiel, Umans, and Viola]{FSUV12}
B.~Fefferman, R.~Shaltiel, C.~Umans, and E.~Viola.
\newblock On beating the hybrid argument.
\newblock In \emph{Proceedings of the 3rd Innovations in Theoretical Computer
  Science Conference}, ITCS '12, pages 468--483, New York, NY, USA, 2012. ACM.
\newblock ISBN 978-1-4503-1115-1.
\newblock \doi{10.1145/2090236.2090273}.
\newblock URL \url{http://doi.acm.org/10.1145/2090236.2090273}.

\bibitem[Feynman(1985)]{F85}
R.~Feynman.
\newblock Quantum mechanical computers.
\newblock \emph{Optics News}, 11\penalty0 (2):\penalty0 11--20, 1985.
\newblock \doi{10.1364/on.11.2.000011}.

\bibitem[Gharibian(2013)]{G13}
S.~Gharibian.
\newblock \emph{Approximation, proof systems, and correlations in a quantum
  world}.
\newblock PhD thesis, University of Waterloo, 2013.
\newblock Available at arXiv:1301.2632 [quant-ph].

\bibitem[Gharibian and Kempe(2012)]{GK12}
S.~Gharibian and J.~Kempe.
\newblock Hardness of approximation for quantum problems.
\newblock In \emph{Automata, Languages and Programming}, volume 7319 of
  \emph{Lecture Notes in Computer Science}, pages 387--398, Berlin, Heidelberg,
  2012. Springer Berlin Heidelberg.
\newblock \doi{10.1007/978-3-642-31594-7_33}.

\bibitem[Gharibian and Sikora(2018)]{GS18}
S.~Gharibian and J.~Sikora.
\newblock Ground state connectivity of local {H}amiltonians.
\newblock \emph{ACM Transactions on Computation Theory}, 10\penalty0 (2), 2018.
\newblock \doi{10.1145/3186587}.

\bibitem[Gharibian et~al.(2015)Gharibian, Huang, Landau, and Woo~Shin]{GHLS15}
S.~Gharibian, Y.~Huang, Z.~Landau, and S.~Woo~Shin.
\newblock Quantum {H}amiltonian complexity.
\newblock \emph{Foundations and Trends in Theoretical Computer Science},
  10\penalty0 (3):\penalty0 159--282, 2015.
\newblock \doi{10.1561/0400000066}.

\bibitem[Gharibian et~al.(2018)Gharibian, Santha, Sikora, Sundaram, and
  Yirka]{GSSSY18}
S.~Gharibian, M.~Santha, J.~Sikora, A.~Sundaram, and J.~Yirka.
\newblock Quantum generalizations of the polynomial hierarchy with applications
  to {QMA(2)}.
\newblock In I.~Potapov, P.~Spirakis, and J.~Worrell, editors, \emph{43rd
  International Symposium on Mathematical Foundations of Computer Science (MFCS
  2018)}, volume 117 of \emph{Leibniz International Proceedings in Informatics
  (LIPIcs)}, pages 58:1--58:16. Dagstuhl Publishing, 2018.
\newblock \doi{10.4230/LIPIcs.MFCS.2018.58}.

\bibitem[Gharibian et~al.(2019)Gharibian, Piddock, and Yirka]{GPY19}
S.~Gharibian, S.~Piddock, and J.~Yirka.
\newblock Oracle complexity classes and local measurements on physical
  {H}amiltonians.
\newblock Available at arXiv:1909.05981 [quant-ph], 2019.

\bibitem[Gill(1977)]{G77}
J.~Gill.
\newblock Computational complexity of probabilistic {T}uring machines.
\newblock \emph{SIAM Journal on Computing}, 6\penalty0 (4):\penalty0 675--695,
  1977.
\newblock \doi{10.1137/0206049}.

\bibitem[Gily\'{e}n and Sattath(2017)]{GS17}
A.~P. Gily\'{e}n and O.~Sattath.
\newblock On preparing ground states of gapped {H}amiltonians: An efficient
  quantum {L}ovász local lemma.
\newblock In \emph{IEEE 58th Annual Symposium on Foundations of Computer
  Science (FOCS)}, pages 439--450, 2017.
\newblock \doi{10.1109/FOCS.2017.47}.

\bibitem[Goldreich(2006)]{G06}
O.~Goldreich.
\newblock On promise problems: {A} survey.
\newblock In O.~Goldreich, A.L. Rosenberg, and A.L. Selman, editors,
  \emph{Theoretical Computer Science}, volume 3895 of \emph{Lecture Notes in
  Computer Science}, pages 254--290. Springer, Berlin, Heidelberg, 2006.
\newblock \doi{10.1007/11685654_12}.

\bibitem[Gosset and Mozgunov(2016)]{GM16}
D.~Gosset and E.~Mozgunov.
\newblock Local gap threshold for frustration-free spin systems.
\newblock \emph{Journal of Mathematical Physics}, 57:\penalty0 091901, 2016.
\newblock \doi{10.1063/1.4962337}.

\bibitem[Gosset et~al.(2017)Gosset, Mehta, and Vidick]{GMV17}
D.~Gosset, J.~C. Mehta, and T.~Vidick.
\newblock {QCMA} hardness of ground space connectivity for commuting
  {H}amiltonians.
\newblock \emph{Quantum}, 1:\penalty0 16, 2017.
\newblock \doi{10.22331/q-2017-07-14-16}.

\bibitem[Hartmanis(1993)]{H87}
J.~Hartmanis.
\newblock Sparse complete sets for {NP} and the optimal collapse of the
  polynomial hierarchy.
\newblock In G.~Rozenberg and A.~Salomaa, editors, \emph{Current Trends in
  Theoretical Computer Science}, pages 403--411. World Scientific, 1993.
\newblock \doi{10.1142/9789812794499_0029}.

\bibitem[Hemachandra(1989)]{H89}
L.~Hemachandra.
\newblock The strong exponential hierarchy collapses.
\newblock \emph{Journal of Computer and System Sciences}, 39\penalty0
  (3):\penalty0 299 -- 322, 1989.
\newblock \doi{10.1016/0022-0000(89)90025-1}.

\bibitem[Hemaspaandra et~al.(1997)Hemaspaandra, Hemaspaandra, and Rothe]{HHR97}
Edith Hemaspaandra, Lane~A. Hemaspaandra, and J\"{o}rg Rothe.
\newblock Exact analysis of {D}odgson elections: {L}ewis {C}arroll's 1876
  voting system is complete for parallel access to {NP}.
\newblock \emph{Journal of the ACM}, 44\penalty0 (6):\penalty0 806--825, 1997.
\newblock \doi{10.1145/268999.269002}.

\bibitem[Karp(1972)]{K72}
R.~Karp.
\newblock Reducibility among combinatorial problems.
\newblock In R.E. Miller, J.W. Thatcher, and J.D. Bohlinger, editors,
  \emph{Complexity of Computer Computations}, The IBM Research Symposia Series,
  pages 85--103. Springer, Boston, MA, 1972.
\newblock \doi{10.1007/978-1-4684-2001-2_9}.

\bibitem[Kempe et~al.(2006)Kempe, Kitaev, and Regev]{KKR06}
J.~Kempe, A.~Kitaev, and O.~Regev.
\newblock The complexity of the local {H}amiltonian problem.
\newblock \emph{SIAM Journal on Computing}, 35\penalty0 (5):\penalty0
  1070--1097, 2006.
\newblock \doi{10.1137/s0097539704445226}.

\bibitem[Kitaev and Watrous(2000)]{KW00}
A.~Kitaev and J.~Watrous.
\newblock Parallelization, amplification, and exponential time simulation of
  quantum interactive proof systems.
\newblock In \emph{Proceedings of the 32nd ACM Symposium on Theory of Computing
  (STOC)}, pages 608--617, 2000.
\newblock \doi{10.1145/335305.335387}.

\bibitem[Kitaev et~al.(2002)Kitaev, Shen, and Vyalyi]{KSV02}
A.~Kitaev, A.~Shen, and M.~Vyalyi.
\newblock \emph{Classical and Quantum Computation}.
\newblock American Mathematical Society, 2002.

\bibitem[Levin(1973)]{L73}
L.~Levin.
\newblock Universal sequential search problems.
\newblock \emph{Problems of Information Transmission}, 9\penalty0 (3):\penalty0
  265--266, 1973.

\bibitem[Liu(2006)]{L06}
Y.~K. Liu.
\newblock Consistency of local density matrices is {QMA}-complete.
\newblock In J.~Díaz, K.~Jansen, J.D.P. Rolim, and U.~Zwick, editors,
  \emph{Approximation, Randomization, and Combinatorial Optimization.
  Algorithms and Techniques}, volume 4110 of \emph{Lecture Notes in Computer
  Science}, pages 438--449. Springer, Berlin, Heidelberg, 2006.
\newblock \doi{10.1007/11830924_40}.

\bibitem[Lockhart and Gonz{\'a}lez-Guill{\'e}n(2017)]{LG17}
J.~Lockhart and C.~E. Gonz{\'a}lez-Guill{\'e}n.
\newblock Quantum state isomorphism.
\newblock Available at arXiv:1709.09622v1 [quant-ph], 2017.

\bibitem[Marriott and Watrous(2005)]{MW05}
C.~Marriott and J.~Watrous.
\newblock Quantum {A}rthur-{M}erlin games.
\newblock \emph{Computational Complexity}, 14\penalty0 (2):\penalty0 122--152,
  2005.
\newblock \doi{10.1007/s00037-005-0194-x}.

\bibitem[Nielsen and Chuang(2011)]{NC00}
M.~A. Nielsen and I.~L. Chuang.
\newblock \emph{Quantum Computation and Quantum Information}.
\newblock Cambridge University Press, 10th edition, 2011.

\bibitem[Oliveira and Terhal(2008)]{OT05}
R.~Oliveira and B.~M. Terhal.
\newblock The complexity of quantum spin systems on a two-dimensional square
  lattice.
\newblock \emph{Quantum Information \& Computation}, 8\penalty0 (10):\penalty0
  900--924, 2008.
\newblock \doi{10.26421/QIC8.10}.

\bibitem[Osborne(2012)]{O11}
T.~J. Osborne.
\newblock Hamiltonian complexity.
\newblock \emph{Reports on Progress in Physics}, 75\penalty0 (2):\penalty0
  022001, 2012.
\newblock \doi{10.1088/0034-4885/75/2/022001}.

\bibitem[Piddock and Montanaro(2017)]{PM17}
S.~Piddock and A.~Montanaro.
\newblock The complexity of antiferromagnetic interactions and {2D} lattices.
\newblock \emph{Quantum Information \& Computation}, 17\penalty0
  (7\&8):\penalty0 636--672, 2017.
\newblock \doi{10.26421/QIC17.7-8}.

\bibitem[Piddock and Montanaro(2018)]{PM18}
S.~Piddock and A.~Montanaro.
\newblock Universal qudit {H}amiltonians.
\newblock Available at arXiv:1802.07130v1 [quant-ph], 2018.

\bibitem[Raz and Tal(2018)]{RT18}
R.~Raz and A.~Tal.
\newblock Oracle separation of {BQP} and {PH}.
\newblock Available at Electronic Colloquium on Computational Complexity
  (ECCC), 2018.

\bibitem[Remscrim(2016)]{R16}
Z.~Remscrim.
\newblock The {H}ilbert function, algebraic extractors, and recursive {F}ourier
  sampling.
\newblock In \emph{IEEE 57th Annual Symposium on Foundations of Computer
  Science (FOCS)}, pages 197--208, Oct 2016.
\newblock \doi{10.1109/FOCS.2016.29}.

\bibitem[Schnoebelen(2003)]{S03}
P.~Schnoebelen.
\newblock Oracle circuits for branching-time model checking.
\newblock In J.~C.M. Baeten, J.~K. Lenstra, J.~Parrow, and G.~J. Woeginger,
  editors, \emph{Automata, Languages and Programming}, volume 2719 of
  \emph{Lecture Notes in Computer Science}, pages 790--801, Berlin, Heidelberg,
  2003. Springer Berlin Heidelberg.
\newblock \doi{10.1007/3-540-45061-0_62}.

\bibitem[Schuch and Verstraete(2009)]{SV09}
N.~Schuch and F.~Verstraete.
\newblock Computational complexity of interacting electrons and fundamental
  limitations of {D}ensity {F}unctional {T}heory.
\newblock \emph{Nature Physics}, 5:\penalty0 732--735, 2009.
\newblock \doi{10.1038/nphys1370}.

\bibitem[Shi and Zhang()]{SZ}
Y.~Shi and S.~Zhang.
\newblock Note on quantum counting classes.
\newblock Available at
  \url{http:://www.cse.cuhk.edu.hk/syzhang/papers/SharpBQP.pdf}.

\bibitem[Valiant and Vazirani(1986)]{VV86}
L.G. Valiant and V.V. Vazirani.
\newblock {NP} is as easy as detecting unique solutions.
\newblock \emph{Theoretical Computer Science}, 47:\penalty0 85--93, 1986.
\newblock \doi{10.1016/0304-3975(86)90135-0}.

\bibitem[Vyalyi(2003)]{Vy03}
M.~Vyalyi.
\newblock {QMA$=$PP implies that {PP} contains {PH}.}
\newblock Available at Electronic Colloquium on Computational Complexity
  (ECCC), 2003.

\bibitem[Watrous(2013)]{W13}
J.~Watrous.
\newblock Quantum computational complexity.
\newblock In R.~Meyers, editor, \emph{Encyclopedia of Complexity and System
  Science}. Springer, 2013.
\newblock \doi{10.1007/978-3-642-27737-5_428-3}.

\bibitem[Yamakami(2002)]{Y02}
T.~Yamakami.
\newblock Quantum {NP} and a quantum hierarchy.
\newblock In Baeza-Yates R., Montanari U., and Santoro N., editors,
  \emph{Foundations of Information Technology in the Era of Network and Mobile
  Computing}, volume~96 of \emph{IFIP — The International Federation for
  Information Processing}, pages 323--336. Springer, Boston, MA, 2002.
\newblock \doi{10.1007/978-0-387-35608-2_27}.

\end{thebibliography}

\appendix
\section{Additional proofs}\label{scn:appendix}

\begin{lemma}\label{l:proof}
   Assume the terminology of Lemma~\ref{l:no}. Then, there exists a valid history state $\ket{\psi'}$ on $W$, $C$, $Q$, and $Q'$ such that $\abs{\brakett{\psi}{\psi'}}^2\geq 1-O(\gamma^{-2})$.
\end{lemma}
\begin{proof}
    Recall that the smallest non-zero eigenvalue of $H_1=\Delta(\hin+\hprop+\hstab)$ is at least $J:=\pi^2\Delta/(64L^3)$ (by Lemma~\ref{l:GKgap}), $\delta=1/\Delta$, and $\eta\geq \max(\snorm{H_2},1)=:K$. Lemma~\ref{cor:kkr} now implies there exists a valid history state $\ket{\psi'}$ satisfying
    \begin{eqnarray*}
        \abs{\brakett{\psi}{\psi'}}^2&\geq&1-\left(\frac{K+\sqrt{K^2+\delta(J-2K)}}{J-2K}\right)^2\\
        &\geq&1-\left(\frac{\eta+\sqrt{\eta^2+\frac{\pi^2}{64L^3}}}{\frac{\pi^2\Delta}{64L^3}-2\eta}\right)^2\\
        &\geq&1-\left(\frac{(1+\sqrt{2})\eta}{\frac{\pi^2\Delta}{64L^3}-2\eta}\right)^2\\
        &=&1-\left(\frac{(1+\sqrt{2})64L^3\eta}{\pi^2\Delta-128L^3\eta}\right)^2\\
        &\geq&1-\Theta\left(\frac{1}{\gamma^2}\right),
    \end{eqnarray*}
    where the second inequality follows since $0\leq K\leq \eta$, the third since $\pi^2/(64L^3)\leq 1$ for $L\geq 1$ and $1\leq \eta^2$, and the last since $\Delta=128L^3\eta\gamma$ and $\gamma\geq 1$.
\end{proof}

\end{document}